\documentclass[12pt,draftcls,onecolumn]{IEEEtran}

\usepackage{graphicx}
\usepackage{url}
\usepackage[pdftex]{color}
\usepackage{amsmath}
\usepackage{amsfonts}
\usepackage{algorithm}
\usepackage{algcompatible}
\usepackage{hyperref}
\usepackage{caption}
\usepackage{subcaption}

\usepackage{amssymb}

\newtheorem{teor}{Theorem}[section]

\newtheorem{propo}{Proposition}[section]
\newtheorem{lemm}{Lemma}[section]

\newtheorem{rem}{Remark}
\usepackage{soul}
\usepackage{tabularx}

\newcommand{\tp}{^{\top}}
\newcommand{\p}{\parallel}
\newcommand{\beq}{\begin{equation}}
	\newcommand{\eeq}{\end{equation}}
\newcommand{\bea}{\begin{eqnarray}}
	\newcommand{\eea}{\end{eqnarray}}

\newcommand{\bsea}{\begin{subeqnarray}}
	\newcommand{\esea}{\end{subeqnarray}}

\def\bmat{\left[ \begin{array}}
	\def\emat{\end{array} \right]}
\DeclareMathOperator{\tr}{tr} %idem
\DeclareMathOperator{\ofd}{ofd}   %{offdiag}
\DeclareMathOperator{\diag}{diag} %idem

\DeclareMathOperator*{\argmin}{arg\,min}

\DeclareMathOperator{\dom}{dom}
\newcommand{\alg}[1]{\begin{align}#1\end{align}}
\newcommand{\Qc}{\mathcal Q}
\newcommand{\Rs}{\mathbb R}
\newcommand{\Zs}{\mathbb Z}

\definecolor{purple}{rgb}{0.62, 0.0, 0.77}

\definecolor{Royalblue}{cmyk}{1,0.30,0.2,0.2}

\begin{document}		
\title{A Robust Approach to ARMA Factor Modeling}

\author{Lucia~Falconi, Augusto~Ferrante, Mattia~Zorzi 
\thanks{L. Falconi, A. Ferrante, and M. Zorzi are with  the Department of Information Engineering, University of Padova, Padova, Italy; e-mail: {\tt\small lucia.falconi@phd.unipd.it} (L. Falconi); {\tt\small augusto@dei.unipd.it} (A. Ferrante); {\tt\small zorzimat@dei.unipd.it} (M. Zorzi).}}

\markboth{DRAFT}{Shell \MakeLowercase{\textit{et al.}}: Bare Demo of IEEEtran.cls for Journals}

\maketitle
	
\begin{abstract}
	This paper deals with the dynamic factor analysis problem for an ARMA process. 
	To robustly estimate the number of factors, we construct a confidence region centered in a finite sample estimate of the underlying model which contains the true model with a prescribed probability. In this confidence region, the problem, formulated as a rank minimization of a suitable spectral density, is efficiently approximated via a trace norm convex relaxation. The latter is addressed by resorting to the Lagrange duality theory, which allows to prove the existence of solutions. Finally, a numerical algorithm to solve the dual problem is presented.
	The effectiveness of the proposed estimator is assessed through simulation studies both with synthetic and real data.
\end{abstract}
	
\begin{IEEEkeywords}
Convex optimization, duality theory, dynamic factor analysis, nuclear norm.
\end{IEEEkeywords}
	
\section{Introduction}

We deal with the problem of constructing a dynamical model
from a high dimensional stream of data that are assumed to be noisy observations
of a process depending on a small number of hidden variables.
In the static case, this problem is known as
{\em factor analysis}. Its origins can be traced back to the beginning 
of the last century and the amount of literature  produced on this topic
is impressive: we refer the readers to  the recent papers \cite{ning2015linear,bertsimas2017certifiably,CFZ_Kybernetika,zorzi2015factor} for 
an overview of the literature and a rich list of references.
The solution of factor analysis problems may be obtained by 
decomposing the covariance matrix of the observed data as the sum
of a diagonal positive definite matrix (accounting for the noise covariance)
and a positive semidefinite matrix whose rank must be as small as possible
since it equals the number of hidden variables in the model. The main problem of this solution is that it is inherently fragile; in fact, even a minuscule variation in the 
covariance matrix of the observed data usually leads to a substantial variation
of the number of hidden variables which is the key feature of the modelling procedure.
On the other hand such a matrix must be estimated and is therefore
subject to errors.
To address this fragility issue, a robust method  has been recently proposed, 
\cite{ciccone2017factor}, that has been generalized with good results also to the dynamic 
framework of learning latent variable dynamic graphical models, \cite{CFZ_TAC19}.

Dynamic Factor Analysis (DFA) has been addressed much more recently, the first contribution in this 
field being apparently \cite{Geweke-77}.
We refer to the surveys  \cite{deistler2007,Stock-Watson-2010} and to the recent paper 
\cite{Figa-Talamanca-et-al} for an overview of the literature on this subject.
In \cite{Bottegal-P-15} an interesting generalization is applied to modelization of dynamical systems.

In this paper, we address the dynamic autoregressive moving average (ARMA) case with the aim of extracting, from the observed data, a model featuring a small number of hidden variables. This is important both from the point 
of view of the model simplicity and to uncover the structure of the mechanism 
generating the data.
The problem may be mathematically formulated as that of decomposing the spectral density
of the process generating the data as the sum of a diagonal spectral density and a low-rank 
one. The fragility issue in this case is even more severe.  
We address the problem as follows:
\begin{itemize}
\item
Given the observed data, we compute by standard methods (e.g. truncated periodogram) a raw estimate 
$\hat{\Phi}$ of the spectral density $\Phi$ generating the data.
\item 
We compute a neighbourhood $\cal N$ of $\hat{\Phi}$ that contains $\Phi$ with prescribed probability;
clearly the size of $\cal N$ depends on the sample size.
\item
We compute a refined estimate $\Phi^\circ \in {\cal N}$ 
by imposing that it admits an additive decomposition as a diagonal spectral density and a spectral density with the lowest possible rank. To this end we set up an optimization problem
that we address by resorting to duality.
In particular, we prove existence of solutions and provide a numerical algorithm to compute
a solution. 
\end{itemize}

Our work may be cast in the rich stream of literature devoted to
learning dynamic models having a topological structure describing the presence or the absence
of interactions among the variables of the systems;  see the former works  
\cite{songsiri2010topology,ARMA_GRAPH_AVVENTI,maanan2017conditional} as well as their extensions to reciprocal processes \cite{Alpago_ELETTERS,Alpago_TAC}, sparse plus low rank graphical models  \cite{ZorzSep,MEManSaid,CFZ_TAC19}, the Bayesian viewpoint proposed in \cite{zorzi_reweight,ZORZI2020109053}
and the case of oriented graphical models \cite{v2020topology,veedu2020topology}.

The contribution of this paper is twofold; first, we propose a procedure to estimate the number of latent factors in dynamic ARMA factor models: this is the most delicate aspect of factor analysis
problems; second, we derive an identification method to estimate the parameters of a factor model describing the observed data. 
%
%The numerical simulations both with synthetic data and with real data give encouraging results.

The outline of the paper is as follows. In Section \ref{Sec:probform} we introduce the DFA problem for moving average (MA) models. Section \ref{Sec:sol} shows that such a problem admits solution by means of duality theory, while Section \ref{Sec:equiv} shows how to reconstruct the solution of the primal problem from the dual one. In Section \ref{Sec:algo} we propose an algorithm to compute the solution of the dual problem. In Section \ref{Sec:ARMAid} we extend the previous ideas to ARMA models. Section \ref{sec:sim} presents some numerical results. Finally, in Section \ref{sec:concl} we draw the conclusions.

\subsection{Notation}
Given a matrix $M$, we denote its transpose by $M^\top$ and by $M_{(i,j)}$ the element of $M$ in the $i$-th row and $j$-th column. If $M$ is a square matrix,  $\tr(M)$, $|M|$ and $\sigma(M)$ denote its trace, its determinant and its spectrum, respectively.
The symbol $\Vert \cdot \Vert$ stands for the Frobenius norm. 
For $A, B \in \Rs^{m \times m}$, we define their inner product as $\langle A,B \rangle:= \tr(A^\top B)$.
Let $\mathbf{Q}_m$ be the space of real symmetric matrices of size $m$; if $M \in \mathbf{Q}_m$ is positive definite or positive semidefinite, then we write $M \succ 0$ or $M \succeq 0$, respectively.
%We deal with Gaussian multivariate processes defined over the integers $\mathbb Z$.
We denote by $(\cdot)^*$ the complex conjugate transpose. $\Phi(e^{i\vartheta})$ for $\vartheta \in [-\pi, \pi] \} $ denotes a function defined on the unit circle $\{e^{i\vartheta}: \vartheta \in [-\pi, \pi] \} $, and the dependence on $\vartheta$ is dropped if necessary. If $\Phi(e^{i\vartheta})$ is positive (semi-)definite $\forall \vartheta \in [-\pi, \pi]  $ we write $\Phi(e^{i\vartheta}) \succ 0$ ( $\succeq 0$).
Integrals are always defined from $-\pi$ to $\pi$ with respect to the normalized Lebesgue measure $d\vartheta /  2 \pi$. 
%We deal with Gaussian multivariate processes defined over the integers $\Zs.$

\section{Identification of MA factor models} \label{Sec:probform}
Consider the MA factor model whose order is $n$:
\alg{\label{MA_FA} y(t)=W_L u(t)+W_D w(t)}
where 
$$ W_L(e^{i\vartheta })=\sum_{k=0}^n W_{L,k}e^{-i\vartheta k}, \quad  W_D(e^{i\vartheta })=\sum_{k=0}^n W_{D,k}e^{-i\vartheta k},$$ 
$W_{L,k}\in \Rs^{m\times r}$, $W_{D,k}\in \Rs^{m\times m}$ diagonal; $u=\{u(t),\; t\in\Zs\}$ and $w=\{w(t),\; t\in\Zs\}$ are normalized white Gaussian noises of dimension 
$r$ and $m$, respectively, such that 
$ \mathbb{E} [u(t) w(s)\tp ] = 0 $ $\forall t,s.$
The aforementioned model has the following interpretation: $u$ is the process which describes the $r$ factors, with  $r \ll m$, not accessible to observation; $W_L$ is the factor loading transfer matrix; $W_L u(t)$ is the latent variable; $W_D w(t)$ is idiosyncratic noise. 
Accordingly,  $y=\{y(t),\; t\in\Zs\}$ is a  $m$-dimensional Gaussian stationary stochastic process with power spectral density 
\alg{\label{LD_decomp}\Phi=\Phi_L+\Phi_D} 
where $\Phi_L=W_L W_L^*\succeq 0$ and $\Phi_D=W_D W_D^*\succeq 0$ belong to the finite dimensional space:
$$ \Qc_{m,n}=\left\{\sum_{k=-n}^n R_k e^{-i\vartheta k}, \; \; R_k=R_{-k}^T \in\Rs ^{m\times m} \right\}. $$
By construction, $\mathrm{rank}(\Phi_L)=r$, where $\mathrm{rank}$ denotes the normal rank (i.e. the rank almost everywhere),  and $\Phi_D$ is diagonal.
Therefore, $y$ represents a factor model if its spectral density can be decomposed as ``low rank plus diagonal'' as in (\ref{LD_decomp}).

Assume to collect a finite length realization of $y$ defined in (\ref{MA_FA}), say $\mathrm y^N=\{\, \mathrm y(1)\ldots \mathrm y(N) \, \}$ where the order $n$ is known. We want to estimate the corresponding factor model, that is the decomposition in (\ref{LD_decomp}) as well as the number of factors $r$. To this aim, given our data $\mathrm y^{N}$, we first compute the sample covariance lags $\hat{R}_j$ as
$$
% \label{hatrj}
\hat{R}_j=\dfrac{1}{N}\sum_{t=0}^{N-j}\mathrm y(t+j)\mathrm y(t)^\top, \; \; \; j=0\ldots n.
$$ 
Then,  an estimate $\hat{\Phi}$ of $\Phi$ is obtained by the truncated periodogram:
\begin{align}\label{pb_per}
	\hat{\Phi}=\sum_{k=-n}^n \hat R_k  e^{i\vartheta k}. \end{align} 
Notice that $\hat \Phi$ could  be not positive definite for all $\vartheta$; in that case, we can add $\varepsilon I_m$ to the right side of Equation \eqref{pb_per}, with the constant $\varepsilon > 0 $ chosen in such a way as to ensure the positivity of $\hat \Phi$. On the other hand, $\hat{\Phi}$ may not admit a low rank plus diagonal decomposition. 
Thus, we  estimate directly the two terms  $\Phi_L$ and $\Phi_D$ of the decomposition 
(\ref{LD_decomp})
by solving  the following optimization problem:
\begin{equation}
	\label{new0}
	\begin{aligned}
		\min_{\Phi,\Phi_L,\Phi_D \in  \mathcal{Q}_{m,n}} & \tr {\int}\Phi_L  \\
		\text{subject to  } &  \Phi_L+\Phi_D= \Phi,\\ 
		& \Phi\succ 0 \text{ a.e.}, \; \; \Phi_L,\Phi_D\succeq  0, \\
		& \Phi_D \hbox{ diagonal},\\
		& \mathcal{S}_{IS}(\Phi|| \hat{\Phi})\leq \delta.
	\end{aligned}
\end{equation} 
Here, the objective function promotes a solution for $\Phi_L$ having low rank, see \cite{ZorzSep}.  The first three constraints impose that  $\Phi_L$ and $\Phi_D$ provide a genuine spectral density decomposition of type 
(\ref{LD_decomp}). The last constraint, in which
$\mathcal{S}_{IS}(\Phi||\hat{\Phi})$ is the Itakura-Saito divergence defined by
$$
%\label{dist_is}
\mathcal{S}_{IS} ( \Phi || \hat \Phi )= \int \log |\hat \Phi {\Phi}^{-1}| +\tr [\hat{\Phi}^{-1} \Phi- I_m],
$$
imposes that $\Phi$ belongs to a set ``centered'' in the nominal spectral density 
$\hat{\Phi}$ and with prescribed tolerance $\delta$. Notice that $\Phi_D$ is uniquely determined by $\Phi$ and $\Phi_L$. Thus, Problem (\ref{new0}) can be rewritten by removing $\Phi_D$: 
\begin{equation}
	\label{new}
	\begin{aligned}
		(\Phi^\circ, \Phi_L^\circ)= \argmin_{\Phi,\Phi_L  \in  \mathcal{Q}_{m,n}} & \tr {\int}\Phi_L  \\
		\text{subject to  }  & \Phi\succ 0 \text{ a.e.}, \; \; \Phi_L,\Phi-\Phi_L\succeq  0, \\
		& \Phi-\Phi_L \hbox{ diagonal},\\
		& \mathcal{S}_{IS}(\Phi|| \hat{\Phi})\leq \delta.
	\end{aligned}
\end{equation} 

%%%%%%%%%%%%%%%%%%%%%%%%%%%%%%%%%%%%%

\subsection{The Choice of $\delta$}\label{Sec:delta}
Before solving our problem, we deal with the choice of the tolerance parameter $\delta$  appearing in the constraint of (\ref{new}). 
This choice should reflect the accuracy of the estimate $\hat \Phi$ of $\Phi$. 
This can be accomplished by choosing a desired probability $\alpha\in (0,1)$ and considering a ball of radius $\delta_\alpha$ (in the Itakura-Saito topology) 
centered in $\hat \Phi$ and containing the true spectrum $\Phi$ with probability $\alpha$. The estimation of $\delta_\alpha$ is not an easy task because we do not know the true power spectral density $\Phi$. Next, we propose a resampling-based method to estimate it. \\
The idea is to approximate $\Phi $ with $\hat{\Phi}$, and use this model to perform a resampling operation. Let $$ W(e^{i\vartheta}) = \sum_{k=0}^n W_{k}e^{-i\vartheta k}, \; W_{k}\in \Rs^{m\times m} $$	be the minimum phase spectral factor of $\hat{\Phi}$ and define the process $\hat{y} =\{\hat{y}(t),\; t\in\Zs\}$ as $\hat{y}(t) := W(e^{i\vartheta}) e(t),$ where  $e(t)$ is an $m$-dimensional normalized white noise. The truncated periodogram (understood as estimator) based on a sample of the process $\hat {y}$ of length $N$ is 
$$ \hat{\mathbf{\Phi}}_r (e^{i\vartheta}) = \sum_{k=-n}^n  e^{-i\vartheta k} \frac{1}{N}\sum_{t=0}^{N-k} \hat y(t+k)\hat y(t)^T , $$
where the subscript  ``$r$'' stands for resampling, as it is the means by which we perform the resampling operation.
By generating a realization $\hat{\mathrm {y}} ^N=\{\, \hat{\mathrm {y}} (1)\ldots \hat{\mathrm {y}} (N)\, \} $ from $\hat{\Phi}$ (i.e. by resampling the data), we can easily obtain a realization of the random variable $\mathcal{S}_{IS}( \hat{\Phi} ||\hat{\mathbf{\Phi}}_r )$. Accordingly, it is possible to compute numerically $\delta_\alpha$ such that $\text{Pr}(\mathcal{S}_{IS}( \hat \Phi || \hat{\mathbf{\Phi}}_r ) \leq \delta_\alpha) = \alpha$ by a standard Monte Carlo procedure.
Numerical simulations show that this technique indeed provides a good estimate of $\delta$.

It is worth noting that if the chosen $\alpha$ is too large with respect to the data length $N$, the resulting $\delta_\alpha$ may be too generous yielding to a diagonal $\Phi$ obeying $\mathcal{S}_{IS}(\Phi|| \hat \Phi)\leq \delta_\alpha$. In this case Problem \eqref{new} admits the trivial solution $\Phi_L=0$ and $\Phi_D=\Phi$.
To rule out this trivial case, $\delta$ in \eqref{new}  must be be strictly smaller than the upper bound 
$$\delta_{\max}:= \min_{\substack{\Phi \in \mathcal{S}_m^+\\ \Phi \hbox{ diagonal}}}  \mathcal{S}_{IS}( \Phi|| \hat{\Phi} )$$ 
where $\mathcal{S}_m^+$ denotes the family of bounded and coercive functions
defined on the unit circle and taking values in the cone of
positive definite $m\times m$ Hermitian matrices.  
%This problem can be easily solved as follows.  
Since $\Phi$ must be diagonal, by denoting with $\phi_i$ and by $\hat{\gamma}_i$ the $i$-th element in the diagonal of $\Phi$ and of $\hat{\Phi}^{-1}$, respectively,  we have
\begin{align*}
	\delta_{\max} &= \left[\sum_{i=1}^{m} \min_{\phi_i \in \mathcal{S}^+_1}  \mathcal{S}_{IS}({\phi}_i || \hat \gamma_i^{-1})\right] +\int \log | \hat \Phi \diag^2(\hat \Phi^{-1})| 
\end{align*}
where $\diag^2(\cdot)$ is the (orthogonal projection) operator mapping  a square matrix $M$ into a diagonal matrix of the same size having the same main diagonal of $M$.
Therefore, since the Itakura-Saito divergence is nonnegative, the solution corresponds to $\phi^{opt}_{i} (e^{i\vartheta})= (\hat{\gamma}_i(e^{i\vartheta}))^{-1}$, $i=1,...,m$ for which $ \mathcal{S}_{IS}( \phi_i^{opt}|| \hat{\gamma}_i^{-1})=0$.
Accordingly,  
\begin{equation}\label{deltamax}
	\delta_{max}  =\int \log | \hat \Phi \diag^2(\hat \Phi^{-1})|.
\end{equation}
The derivation of the aforementioned result is based on reasonings similar to \cite[Section IV]{CFZ_TAC19}.

A more generous upper bound can be derived by assuming that $\Phi$ is the spectrum of an MA process of order $n$.
However, numerical experiments showed that $\delta_{max}\gg \delta_{\alpha}$ even in the case that $N$ is relatively small.

%%%%%%%%%%%%%%%%%%%%%%%%%%%%%%%%%
\section{Problem solution} \label{Sec:sol}
In this section we first provide a finite dimensional matrix parametrization of Problem (\ref{new}). The latter is then analyzed by resorting to the Lagrange duality theory, which allows us to prove the existence of a solution.

\subsection{Matricial Reparametrization of the Problem}\label{section_pb_formulation}
To study Problem \eqref{new} it is convenient to introduce the following matrix parametrization for $\Phi, \Phi_L$ and $\Phi-\Phi_L$:
\begin{equation}\label{matrepofspectra}
	\begin{aligned}
		\Phi= \Delta X\Delta^* \; & \in \mathcal{Q}_{m,n}\\
		\Phi_L= \Delta L\Delta^*\; & \in \mathcal{Q}_{m,n}\\
		\Phi - \Phi_L= \Delta (X-L) \Delta^* \; & \in \mathcal{Q}_{m,n}\\
	\end{aligned}
\end{equation}
where $\Delta(e^{i\theta})$ is the so-called shift operator:
\beq\label{defDelta}
\Delta(e^{i\vartheta}):= [I_m \quad e^{i\vartheta}I_m  \quad \dots \quad e^{in\vartheta}I_m];
\eeq
$X$ and $L$ are matrices in $\mathbf{Q}_{m(n+1)}$  
and $X_{ij}$ denotes the block of $X$ in position $i,j$ with $i,j=0,\dots,n$, so that
\[X=\begin{bmatrix}
	X_{00}       & X_{01}  & \dots & X_{0n} \\
	X_{01}\tp       & X_{11}  & \dots & \vdots \\  
	\vdots  & \vdots & \vdots & \vdots \\ 
	X_{0n}\tp       & X_{1n}\tp & \dots  & X_{nn}
\end{bmatrix}.\] Moreover, $\textbf{M}_{m,n}$ denotes the vector space of matrices of the form 
\beq\label{formulaperyinmmn}
Y:=[Y_0 \quad Y_1 \quad ... \quad Y_n],\quad Y_0 \in \mathbf{Q}_m,\ \ Y_1, ...,Y_n \in \mathbb{R}^{m\times m}.
\eeq
The linear mapping $T:\mathbf{M}_{m,n}\rightarrow \mathbf{Q}_{m(n+1)}$  constructs a symmetric 
block-Toeplitz matrix from its first block row so that if $Y$ is given by (\ref{formulaperyinmmn}),
\[T(Y)=\begin{bmatrix}
	Y_{0}       & Y_{1} 
	%& Y_{2} 
	& \dots & Y_{n} \\
	Y_{1}\tp       & Y_{0} 
	%& Y_{1} 
	& \ddots & \vdots \\
	%Y_{2}\tp & Y_1\tp & Y_{0} & \ddots & \vdots \\
	\vdots  & \ddots 
	%& \ddots 
	& \ddots & Y_{1}\\ 
	Y_{n}\tp       & \dots 
	%& \dots 
	& Y_{1}\tp & Y_{0}
\end{bmatrix}.\] The adjoint of $T$ is the mapping $D:\mathbf{Q}_{m(n+1)}\rightarrow \mathbf{M}_{m,n}$ defined by $D(X)=[[D(X)]_0\quad \dots \quad [D(X)]_n]$ with
\[[D(X)]_0= \sum_{h=0}^{n}X_{hh}, \quad [D(X)]_j=2\sum_{h=0}^{n-j} X_{h h+j}, \  j=1,...,n.\]

Next, the objective is to provide a  more convenient formulation of Problem \eqref{new} in terms of $X$ and $L$.
To this end, we have to take into account the following points.

\textit{ 1) Positivity Constraints $\Phi \succ 0\text{ a.e.} $  and $\Phi_L, \Phi-\Phi_D \succeq 0:$ }
It can been shown (see, for example, \cite[Appendix A]{ZorzSep}) that, for any $\Psi \in \mathcal{Q}_{m,n}$, $\Psi\succeq 0$ if and only if there exists a matrix $P\in \mathbf{Q}_{m(n+1)}$ such that $\Delta P \Delta^*$ and $P\succeq 0$. Therefore, we replace the conditions $\Phi_L \succeq 0$ with $L\succeq 0$, the condition $\Phi-\Phi_L \succeq 0$ with $X -L\succeq 0$. Note that these conditions only guarantees $X\succeq 0$ and thus $\Phi$ to be positive semidefinite, however we will show that this is sufficient to guarantee that $\Phi \succ 0 \text{ a.e.}$ at the optimum.

\textit{2) Constraint $\Phi-\Phi_L$ diagonal:}
Let $\ofd\,:\, \Rs^{m\times m}\rightarrow \Rs^{m\times m}$ denote the linear operator such that, given $A\in \Rs^{m\times m}$, $\ofd(A)$ is the matrix in which each off-diagonal element is equal to the corresponding element of $A$ and each diagonal element is zero. We define the ``block ofd'' linear operator $\ofd_B \,: \,\mathbf{M}_{m,n} \rightarrow \mathbf{M}_{m,n}$ as follows. Given $Z=[\, Z_0 \; Z_1 \ldots Z_n\,]\in\mathbf{M}_{m,n}$, then
$$ \ofd_B(Z)=[\, \ofd(Z_0) \; \ofd(Z_1) \ldots \ofd(Z_n)\,]. $$ 
It is not difficult that $\ofd_B$ is a self-adjoint operator, since $\ofd$ is self-adjoint as well. 
Then, it is easy to see that the condition $\Phi-\Phi_L$ diagonal is equivalent to the condition $ [D(X - L)]_j$ diagonal for $j=0,\dots,n$, that is $\ofd_B(D(X-L))=0$.

\textit{3) The Low Rank Regularizer:} We have
\begin{align*}\tr\int \Phi_L  &=\tr\int \Delta L\Delta^* 
	= \tr \left( L\int \Delta^*\Delta \right) = \tr(L)
\end{align*}
where we exploited  the fact that
$\int e^{ij\vartheta}= 1$ if $j=0$, and $\int e^{ij\vartheta}= 0$ otherwise.

\textit{4) The Divergence Constraint:}
A convenient matrix parameterization of the Itakura-Saito divergence $\mathcal{S}_{IS}( \Phi || \hat\Phi)$ can be obtained by making use of the following facts.

First, since $\Phi=\Delta X \Delta^*$ with $X\succeq 0$, there exists $A\in \mathbb{R}^{m\times m(n+1)}$ such that $X=A^\top A$. Then, by using the Jensen-Kolmogorov formula we obtain
\beq \label{eq::JKformula}
\int  \log|\Phi| = \!\! \int  \log |\Delta A^\top A\Delta^*|  = \log |A_0^\top A_0|=\log|X_{00}| \eeq
which holds provided that $X_{00}\succ 0$ and  $\Phi$ is coercive (i.e. $|\Phi|$ is bounded away from zero on the unit circle).
We need to generalize this result to spectral densities that may be singular on the unit circle.
This is possible because the zeros of a rational spectral density, if any, have finite multiplicity so that the logarithm of the determinant of a rational spectral $\Phi$ is integrable
as long as the normal rank of $\Phi$ is full.
\begin{lemm} \label{lemm:Jensen-Kolmogorov}
Consider a power spectral density $\Phi \in \Qc _{m,n}$ having full normal rank. 
Let $X \in \mathbf{Q}_{m(n+1)}$ be such that  $X\succeq 0$,  $X_{00}\succ 0$, and $\Phi = \Delta X \Delta^* $. 
Then $$ \int  \log|\Phi| =\log|X_{00}|. $$
\end{lemm}
The proof is deferred to the appendix.

A second observation in order to conveniently parameterize the Itakura-Saito divergence constraint is that, by exploiting the cyclic property of the trace, 
\begin{align*}
	\int \tr ( \hat{\Phi}^{-1} \Phi)  &= \int \tr (\hat{\Phi}^{-1} \Delta X \Delta^* )\\
	& = \tr \left(  X \int  \Delta^* \hat\Phi^{-1} \Delta \right) = \langle  X, T(\hat{P})\rangle,
\end{align*}
where $\hat{P}$ is defined from the expansion $$ \hat \Phi^{-1}=\sum_{k=-\infty}^\infty \hat P_k e^{-i\vartheta k} $$
as $\hat{P}:=[\hat{P}_0\dots \hat{P}_n].$
%By construction, $\int  \Delta^* \hat{\Phi}^{-1} \Delta  =   T(\hat{P}). $

Summing up, we get the following matrix re-parametrization of Problem \eqref{new}:
\begin{equation}
	\label{problem}
	\begin{aligned}
		(X^\circ, L^\circ)= \argmin_{X,L\in \mathbf{Q}_{m(n+1)}}  & \tr(L) \\
		\text{subject to } &  X_{00}\succ 0, \, L\succeq 0, \, X-L\succeq 0, \\ 
		& \ofd_B(D[X-L])=0, \\
		& -\log|X_{00}|+\int \log|\hat{\Phi}| \\& + \langle X, T(\hat{P})\rangle - m\leq \delta.
	\end{aligned}
\end{equation}
We remark once again that  to prove the equivalence between \eqref{new} and \eqref{problem} we still need to show that
$\Phi\succ 0 \text{ a.e.}$ at the optimum: this fact will be established after the variational analysis.

%%%%%%%%%%%%%%%%%%%%%%%%%%%%%%%%%%%%
\subsection{The Dual Problem}\label{Sec:dual}
We reformulate the constrained minimization problem in \eqref{problem} as an unconstrained problem by means of Duality Theory.\\
If we use $V, U\in\mathbf{Q}_{m(n+1)}, \; V,U\succeq 0$ as the multipliers associated with the constraints on the positive semi-definiteness of $X-L$ and $L$, respectively; 
$Z\in \mathbf{M}_{m,n}$ as the multiplier associated with the constraint $\ofd_B(D(X-L))=0$ and 
$\lambda \in \mathbb{R}, \lambda\geq 0$, as the multiplier associated with the Itakura-Saito divergence, 
then the Lagrangian of Problem \eqref{problem} is
\begin{equation}
	\label{lagrangian}
	\begin{aligned}
		\mathcal{L}(X,L,\lambda, &U, V, Z) = \tr(L)- \langle V, X-L\rangle -\langle U,  L \rangle+ \\
		&  \langle Z, \ofd_B(D(X-L) )\rangle + \lambda \big(-\log |X_{00}| +\\
		&  \int \log|\hat{\Phi}| +  \langle X, T(\hat{P})\rangle - m - \delta \big)\ \\
		& = \langle L, I - U + V-T(\ofd_B(Z))\rangle + \\
		&  \langle X, T(\ofd_B(Z))-V+\lambda T(\hat{P})\rangle -\\
		&   \lambda \big(\log |X_{00}|  -\int \log|\hat{\Phi}| + m + \delta \big). 
	\end{aligned}
\end{equation}
Note that we have not included the constraint $X_{00}\succ 0$ because, as we will show later on, this condition is automatically met by the solution of the dual problem.

The dual function is defined as the infimum of $\mathcal{L}$ over $X$ and $L$. 
Thanks to the convexity of the Lagrangian, we rely on standard variational methods to
characterize the minimum.
\begin{itemize}
	\item \textit{Partial minimization with respect to $L$:} $\mathcal{L}$ depends on $L$ only through $\langle L, I - U + V- T(\ofd_B(Z))\rangle$ which is bounded below only if
	\begin{equation}
		\label{over_S}
		I-U+V-T(\ofd_B(Z))=0.
	\end{equation}
	Thus, we get that
	\begin{equation*}
		\begin{aligned}
			\inf_{L} \mathcal{L}= \begin{cases}
				\langle X, T(\ofd_B(Z))-V + \lambda T(\hat{P})\rangle - \\
				\lambda \big(\log |X_{00}| -\int \log|\hat{\Phi}| 
				+ m + \delta \big)  & \text{if \eqref{over_S}}\\
				-\infty & \text{otherwise.} 
			\end{cases}
		\end{aligned}
	\end{equation*}
	\item \textit{Partial minimization with respect to $X$:} The terms in $X_{00}$ are bounded below only if 
	\begin{equation}
		\label{over_X0}
		\left[ T(\ofd_B(Z))-V + \lambda T(\hat{P}) \right]_{00}\succ 0
	\end{equation}
	and are minimized if $\lambda>0$ and
	\beq \label{X00eq}
	X_{00}=\left(\big[ T(\hat{P}) + \lambda^{-1}(T(\ofd_B(Z))-V) \big]_{00}\right)^{-1}.
	\eeq
	The Lagrangian is linear in the remaining variables $X_{lh}$, for $(l,h)\neq (0,0)$, and therefore bounded below only if
	\begin{equation}
		\label{over_X}
	\left[ T(\ofd_B(Z))-V + \lambda T(\hat{P})\right]_{lh} = 0 \quad \forall (l,h)\neq (0,0).
	\end{equation}
	Therefore, the minimization of the Lagrangian with respect to $X$ and $L$ is finite if and only if
	\eqref{over_S}, \eqref{over_X0}, and \eqref{over_X} hold in which case
	\begin{equation*}
		\begin{aligned}
			\min_{X,L} \mathcal{L}= 
			- \lambda \big(- \log \big| \big[ T(\hat{P})+\lambda^{-1}(T(\ofd_B(Z))\\ 
			-V)\big]_{00}\big|  -\int \log\big|\hat{\Phi}\big| + \delta \big).  
		\end{aligned}
	\end{equation*}
	Otherwise the Lagrangian has no minimum and its infimum is
	$-\infty $.
\end{itemize}

To simplify the notation, let us define the vector space $\mathcal{O}$ as:
\begin{align*}
	\mathcal{O}:=\lbrace Z \in \mathbf{M}_{m,n}:\ofd_B(Z)=Z, \; j=0,...,n \rbrace;
\end{align*} 
since $Z$ always appears in the form $\ofd_B(Z)$, we can replace it with $Z\in\mathcal O$. Then, we can formulate the dual problem for the Lagrangian \eqref{lagrangian} as
\begin{equation}
	\label{dual0}
\underset{(\lambda, U, V, Z) \in\tilde{\mathcal{C}}}{ \max} \tilde{J}
\end{equation}
where
\[\tilde{J}:=\lambda \Big( \log \big|\big[ T(\hat{P})+ \lambda^{-1}(T(Z) -V) \big]_{00}\big| 
+\int \log|\hat{\Phi}|  - \delta \Big) \]
and the feasible set $\tilde{\mathcal{C}}$ is given by:  
\begin{align*}
	\tilde{\mathcal{C}}& := \lbrace (\lambda, U, V, Z): U, V \in \mathbf{Q}_{m(n+1)}, U, V \succeq 0, Z \in \mathcal{O}, \\ 
	& \lambda \in \mathbb{R}, \lambda >0, I-U+V-T(Z)=0,
	[ \lambda T(\hat{P})+ T(Z ) - \\ & V]_{00} \succ 0, 
	[  \lambda T(\hat{P})+ T(Z )-V]_{lh} = 0 \; \; \forall (l,h)\neq (0,0) \rbrace.
\end{align*}
Note that the constraints $I-U+V-T(Z)=0$ and $U\succeq 0$ are equivalent to the constraint $I+V-T(Z)\succeq 0$. Thus, we can eliminate the redundant variable $U$; 
moreover, by changing the sign to the objective function $\tilde{J} $ and observing that $ \big[  T(\hat{P})+ \lambda^{-1} (T(Z )-V) ]_{00}=\hat{P}_0+\lambda^{-1}(Z_0-V_{00})$, we can rewrite \eqref{dual0} as a minimization problem:
\begin{equation}
	\label{dual}
	\min_{(\lambda, V, Z) \in \mathcal{C}} J
\end{equation}
where 
\[
J:=  \lambda \Big(- \log \big| \hat{P}_0+ \lambda^{-1}(Z_0 -V_{00}) \big| 
-\int \log|\hat{\Phi}|  + \delta \Big). \]
and the corresponding feasible set $ \mathcal{C}$ is:
\begin{align*}
	\mathcal{C}& := \lbrace  ( \lambda,V, Z):  V \in \mathbf{Q}_{m(n+1)},  V \succeq 0, Z \in \mathcal{O},  \\&
	 I+ V - T(Z)\succeq 0, 	\lambda \in \mathbb{R}, \lambda >0, 	[ \lambda \hat{P}_0+Z_0-V_{00} ] \succ 0, \\& 
	 [\lambda(T(\hat{P}))+  T(Z)-V ]_{lh} = 0 \; \; \forall (l,h)\neq (0,0) \rbrace.
\end{align*}
%Next we address existence of solutions to \eqref{dual}.

%%%%%%%%%%%%%%%%%%%%%%%%%%%%%%%%%%%%%%%%%%%%%%

\subsection{Existence of solutions} \label{Sec::exist}

The aim of this section is to show that \eqref{dual} admits solution.
The set  $ \mathcal{C}$ is not compact, as it is neither closed nor bounded. 
We show that we can restrict the search of the minimum of $J$ over a compact set.
Then, since the objective function is continuous over  $ \mathcal{C}$ (and hence over the restricted compact set), 
we can use Weierstrass's Theorem to conclude that the problem does admit a minimum.

The first step consists in showing that we can restrict  $ \mathcal{C}$ to a subset where $\lambda \geq \varepsilon$ with $ \varepsilon > 0$ a positive constant.
\begin{propo} \label{prop::first_restriction}
	Let $ ( \lambda^{(k)} , V^{(k)} , Z^{(k)} )_{k \in \mathbb{N} } $ be a sequence of elements in  $\mathcal{C}$ such that 
	$$	\lim_{k \to \infty} \lambda^{(k)} = 0 . $$
	Then, such a sequence cannot be an infimizing sequence.
\end{propo}
The proof is essentially the same as the proof of Proposition 6.1 in \cite{CFZ_TAC19} and it is therefore omitted.\\
As a consequence, minimizing the dual functional over the set  $\mathcal{C}$ is equivalent to minimize it over the set:
\begin{align*}
	\mathcal{C}_1 &:=  \lbrace  ( \lambda,V, Z ) : V \in \mathbf{Q}_{m(n+1)}, V \succeq 0, Z \in \mathcal{O}, \\& I+V-T(Z)\succeq 0,   \lambda \in \mathbb{R}, \lambda \geq \varepsilon,  
	[ \lambda \hat{P}_0+Z_0-V_{00} ] \succ 0,\\& [\lambda(T(\hat{P}))+  T(Z)-V ]_{lh} = 0 \; \forall (l,h)\neq (0,0) \rbrace.
\end{align*}

Next we show that we can restrict $\mathcal{C}_1$ to a subset in which both $(T(Z) - V)$ and $\lambda$ cannot diverge.
\begin{propo} \label{prop::second_restriction}
	Let $ ( \lambda^{(k)} , V^{(k)} , Z^{(k)} )_{k \in \mathbb{N} } $ be a sequence of elements in  $\mathcal{C}_1$ such that either
	$$ 	\lim_{k \to \infty} \p T(Z^{(k)}) - V^{(k)} \p  = +\infty $$ 
	or  
	$$ 		\lim_{k \to \infty} \lambda^{(k)}  = + \infty  $$
	or both. Then, such a sequence cannot be an infimizing sequence.
\end{propo}
The above result is obtained by following arguments similar to the proof of Proposition 6.2 in \cite{CFZ_TAC19} with a few small differences; we refer the interested reader to \cite[Appendix C]{Tesi_Falconi} for the detailed proof. \\
It follows from the previous proposition that there exists $\beta \in \mathbb{R}$ with $ \mid \beta \mid < \infty$ such that 
$
T(Z) - V \succeq \beta I,
$
and $ 0 < \gamma  < \infty $ such that $ \lambda \leq \gamma $. Therefore,  the set $\mathcal{C}_1$ can be further restricted to the set: 
\begin{align*}
	\mathcal{C}_2   &:= \lbrace 
	( \lambda,V, Z):  V \in \mathbf{Q}_{m(n+1)},  V \succeq 0, Z \in \mathcal{O},  \lambda \in \mathbb{R}, \\ & \beta I \preceq T(Z) - V \preceq I,  \gamma \geq \lambda \geq \varepsilon, [ \lambda \hat{P}_0+Z_0-V_{00} ] \succ 0,  \\ & [\lambda(T(\hat{P}))+  T(Z)-V ]_{lh} = 0 \; \forall (l,h)\neq (0,0) \rbrace.
\end{align*}

In addition, it is not possible for $V$ and $Z$ to diverge while keeping  the difference $ T(Z) - V $ finite.
Accordingly, we can further restrict the search for the optimal solution to a subset $\mathcal{C}_3$ in which neither $V$ nor $Z$ can diverge: 
\begin{propo} \label{prop::third_restriction}
	Let $ ( \lambda^{(k)} , V^{(k)} , Z^{(k)} )_{k \in \mathbb{N} } $ be a sequence of elements in  $\mathcal{C}_2$ such that
	\begin{equation} \label{eq::V_divergence}
		\lim_{k \to \infty} \p V^{(k)} \p = +\infty  
	\end{equation}
	or 
	\begin{equation} \label{eq::Z_divergence}
		\lim_{k \to \infty} \p Z^{(k)} \p = +\infty 
	\end{equation}
	or both. Then, such a sequence cannot be an infimizing sequence. 
\end{propo}
The proof can be found in the appendix. \\
Thus, the minimization over $\mathcal{C}_2$ is equivalent to the minimization over the subset: 
\begin{align*}
	\mathcal{C}_3  &:= \lbrace 
	( \lambda,V, Z ): V \in \mathbf{Q}_{m(n+1)}, \alpha I \succeq V \succeq 0, Z \in \mathcal{O}, \lambda \in \mathbb{R},
	\\& \beta I \preceq T(Z) - V \preceq I,       \gamma \geq \lambda \geq \varepsilon, [ \lambda \hat{P}_0+Z_0-V_{00} ] \succ 0,
	\\&  [\lambda(T(\hat{P}))+  T(Z)-V ]_{lh} = 0 \; \forall (l,h)\neq (0,0) \rbrace 
\end{align*}
for a certain $\alpha > 0$ positive constant.

Finally, we consider a sequence $ ( \lambda^{(k)} , V^{(k)} , Z^{(k)} )_{
	k \in \mathbb{Z} } \in \mathcal{C}_3 $ such that  $ [ ( \lambda^{(k)} )^{-1} \big( [Z^{(k)}]_0 - [V^{(k)}]_{00} \big) + \hat{P_0} ] $ tends to be singular as $k \to \infty$. This implies that $
| ( \lambda^{(k)} )^{-1} \big( [Z^{(k)}]_0 - [V^{(k)}]_{00} \big) + \hat{P_0} | $ tends to zero and hence $J \to + \infty$. Thus, such a sequence cannot be an infimizing sequence. 
Therefore, the final set $\mathcal{C}_C$ is:
\begin{align*}
	\mathcal{C}_C  & := \lbrace 
	(\lambda,V, Z): V \in \mathbf{Q}_{m(n+1)}, \alpha I \succeq V \succeq 0, Z \in \mathcal{O},  
	  \beta I   \preceq \\& T(Z) - V \preceq I, \lambda \in \mathbb{R},  \gamma \geq \lambda \geq \varepsilon,  [ \lambda \hat{P}_0+Z_0-V_{00} ] \succeq \mu I, \\&[\lambda(T(\hat{P}))+  T(Z)-V ]_{lh} = 0 \; \forall (l,h)\neq (0,0) \rbrace
\end{align*}
where $\alpha, \beta, \gamma, \varepsilon$ and $\mu$ such that $ |\alpha|, | \beta|, | \gamma |, | \varepsilon |$ and $| \mu | < + \infty. $
\begin{teor}
	Problem \eqref{dual} is equivalent to 
	\begin{equation*}
		\min_{(\lambda, V, Z) \in \mathcal{C}_C} J(\lambda, V, Z)
	\end{equation*}
	and it admits solution.
\end{teor}
\begin {proof}
Equivalence of the two problems has already been proven by the previous arguments. Since $\mathcal{C}_C$ is closed and bounded, hence compact, and $J$ is continuous over $\mathcal{C}_C$,
by the Weierstrass’s Theorem the minimum exists. $\; \blacksquare$
\end{proof}

%%%%%%%%%%%%%%%%%%%%%%%%%%%%%%%%%%%%%%%%%%%%%%%%%%%

\section{Solution of the primal problem}\label{Sec:equiv}
In this section, after proving that the primal problem (\ref{new}) and its matrix reformulation (\ref{problem}) are equivalent, we show how to recover the solution of the primal problem.

Let $(\lambda^\circ,V^\circ,Z^\circ)$ be a solution of (\ref{dual}) and $(X^\circ,L^\circ)$ be the corresponding solution of (\ref{problem}). 
Since $X_{00}^\circ$ is positive definite,  $\log|X_{00}^\circ|$ is finite.  
By Lemma \ref{lemm:Jensen-Kolmogorov}, at the optimum $\int  \log|\Phi|$ must be finite as well; this implies that 
$ \Phi (e^{i \vartheta}) $, $ \vartheta \in  \left[ - \pi , + \pi \right], $  may be singular at most on a set of zero measure,
or, in other terms, $\Delta X^\circ \Delta^*  \succ 0 \text{ a.e.}$. This observation leads to the following proposition:
\begin{propo}
\label {ProblemsEquivalence}
Let $(X^\circ,L^\circ)$ be a solution of (\ref{problem}). Then  $\Delta X^\circ \Delta^*  \succ 0 \text{ a.e.}  $. Accordingly, \eqref{new} and \eqref{problem} are equivalent.
\end{propo}

Now we are ready to show how to recover the solution of the primal problem; to this aim we need the following result, see \cite{songsiri2010graphical}.

\begin{lemm}\label{LemmaYW}
Let $Z\in \mathbf M_{m,n}$ and $W\in \mathbf{Q}_{m}$. If $W\succ 0$ is such that 
\alg{ T(Z)\succeq \left[\begin{array}{cc}W & 0  \\0 & 0 \end{array}\right]} then $ T(Z)\succ 0$.
\end{lemm}

Exploiting the constraints $[\lambda(T(\hat{P}))+  T(Z)-V ]_{lh} = 0, \forall (l,h)\neq (0,0)  $ and  $ [ \lambda \hat{P}_0+Z_0-V_{00} ] \succ 0 $, 
it is not difficult to see that 
\alg{ \label{eq:def_V_opt} V^\circ= \lambda^\circ  T(\hat P)+T(Z^\circ)- \left[\begin{array}{cc}W^\circ & 0  \\0 & 0 \end{array}\right]} where 
\alg{  \label{eq:def_W_opt}W^\circ:= Z^\circ_{0}-V^\circ_{00}+\lambda^\circ  \hat P_0\succ 0.} 
Since $V^\circ \succeq 0$ and in view of Lemma \ref{LemmaYW},  $ \lambda^\circ T(\hat P) +T(Z^\circ) \succ 0$. Hence, $V^\circ $
has rank at least equal to $mn$.

Since the duality gap between (\ref{problem}) and (\ref{dual}) is equal to zero, we have that 
$ \langle V^\circ,X^\circ -L^\circ\rangle=0$,
which in turn implies 
\alg{\label{d_Gap_1} V^\circ (X^\circ -L^\circ ) = 0 }
because $V^\circ, X^\circ-L^\circ\succeq 0$.
Recalling that $ \mathrm{rank} ( V^\circ ) \geq mn,$ in view of (\ref{d_Gap_1}) the matrix $X^\circ-L^\circ$ has rank at most equal to $m$. 
Let $\mathrm{rank} (X^\circ-L^\circ) =\tilde{m} \leq m$. 
Then, there exists a full-row rank matrix  $A\in \mathbb{R}^{\tilde{m}\times m(n+1)}$ such that 
\beq \label{eq:AtA} X^\circ-L^\circ=A^\top A.
\eeq 
By (\ref{d_Gap_1}), it follows that 
$  V^\circ A^\top =0. $
Let $Y_D := [v_o\quad v_1 \quad ... \quad v_l] \in \mathbb{R}^{ m(n+1) \times l}$ denote the matrix whose columns form a basis of $\ker(V^\circ)$. 
Note that the dimension $l$ of the null space of $ V^\circ $ is at least $\tilde{m}$ because  $ \mathrm{Im} (A^\top) \subseteq \ker (V^\circ)$ and $ \mathrm{rank} (A^\top) = \tilde{m} $; also  $ l \leq m  $ because $ \mathrm{rank} (V^\circ) \geq mn$. Rewriting the matrix $A^\top$ as
$ A^\top = Y_{D} S $
with $S \in \mathbb{R}^{l \times \tilde m}$, from \eqref{eq:AtA} we obtain
\beq \label{D_opt}  X^\circ-L^\circ = Y_{D} Q_D Y_{D}^\top, \eeq
with $Q_D := S S^\top \in  \mathbf{Q}_{l}$ unknown. 

In a similar fashion, by the zero duality gap between (\ref{problem}) and (\ref{dual}), 
the \textit{complementary slackness condition} for the multiplier associated to the positive semi-definiteness of $L$ reads as
$ \langle U^\circ, L^\circ\rangle=0,$
which in turn implies
$ U^\circ  L^\circ  = 0.$
Repeating the same reasoning as before, it can be seen that, if the dimension of the null space of ${U^\circ}$ is $\tilde{r}$ with $\tilde{r} \geq r$  and
$ Y_L := [u_o\quad u_1 \quad ... \quad u_{\tilde{r}}] \in \mathbb{R}^{m(n+1) \times \tilde{r}}$ is 
a matrix whose columns form a basis of $\ker(U^\circ)$, then $L^\circ$ can be written as
\alg{ \label{L_opt} L^\circ =Y_{L} Q_L Y_{L}^\top }
with $Q_L  \in  \mathbf{Q}_{\tilde{r}}$ unknown.
Plugging (\ref{L_opt}) into (\ref{D_opt}), we then obtain
\alg { \label{sys_1} X^\circ - Y_{L} Q_L Y_{L}^\top = Y_{D} Q_D Y_{D}^\top .}
Assume now that each block of $ X^\circ - L^\circ $ is diagonal, namely
\alg { \label{ofd_D} \ofd ( \left[Y_{D} Q_D Y_{D}^\top \right]_{hk} ) = 0 \quad h,k = 0, ..., n  .}
\begin{rem}
We can make the previous assumption without loss of generality.
Indeed, let $(\Phi^\circ$, $\Phi_L^\circ)$ be the solution of Problem \eqref{new} and $\Phi_D^\circ = \Phi^\circ - \Phi_L^\circ$;  $X$, $L$ and $D = X - L$ are any matrices in $\mathbf{Q}_{m(n+1)}$ such that  $\Phi^\circ = \Delta X \Delta^*$, $\Phi^\circ_L = \Delta L \Delta^*$ and $\Phi^\circ_D = \Delta D \Delta^*$. 
We can always consider a different matrix parametrization $ (\tilde{X}, \tilde{L}, \tilde{D} ) $ for $\Phi^\circ$, $\Phi^\circ_L$ and $\Phi^\circ_D$  as follows.
First notice that there always exists a matrix $\tilde{D}$ with all diagonal blocks such that $ \Phi^\circ_D = \Delta \tilde{D} \Delta^*$; in other words, we can always find $\delta D \in \mathbf{Q}_{m(n+1)}$ such that $ \Delta \delta D  \Delta^* = 0 $ and $\tilde{D} := D + \delta D $ satisfies $\ofd ( \big[ \tilde{D} \big]_{hk} ) = 0$  for $ h,k = 0, ..., n. $ Now, let $\delta X \in \mathbf{Q}_{m(n+1)}$ such that $ \Delta \delta X  \Delta^* = 0 $ and $ \tilde{X} := X +  \delta X $ satisfies \eqref{X00eq}. Define $ \tilde{L} = \tilde{X} - \tilde{D} = X - D  + \delta L $ with $ \delta L := \delta X - \delta{D} $.
It is easy to see that $\Phi^\circ = \Delta \tilde{X} \Delta^*$ and $\hat{\Phi}_L = \Delta \tilde{L} \Delta^*$. This means that  $(\tilde{X}, \tilde{L} )$ is still a solution of Problem \eqref{problem} and it allows us to restrict to solutions of  (\ref{problem}) for which \eqref{ofd_D} holds. 
\end{rem}

By applying the $\ofd$ operator to both sides of (\ref{sys_1}) and exploiting the assumption \eqref{ofd_D},
it is not difficult to obtain:
\beq \label{sys_L} 
\ofd ( \left[Y_{L} Q_L Y_{L}^\top \right]_{00} ) = \ofd (X^\circ _{00}) \eeq
which is a  system of $m(m-1) /2$ linear equations in the $\tilde{r}(\tilde{r}+1)/2$ unknowns $ Q_L $.  Notice that $X_{00}$ is given by (\ref{X00eq}).
Finally, once $L^\circ$ is computed, in order to retrieve $Q_D$ we exploit \eqref{ofd_D} and the following system of $m(m+1)/2$ linear equations: 
\alg { \label{sys_D}\left[Y_{D} Q_D Y_{D}^\top \right]_{00}  =   X^\circ _{00} - L^\circ _{00} . }
Since both the dual and the primal problem admit solution, the resulting systems of equations \eqref{ofd_D}, \eqref{sys_L} and \eqref{sys_D} do
admit solutions. %in $Q_L$ and $Q_D$.

\section{The proposed algorithm}\label{Sec:algo}
In this section we propose an algorithm to solve numerically the dual problem. To start with, as observed in Section \ref{Sec:equiv}, we rewrite \eqref{dual} in a different fashion by getting rid of the slack variable
$V \in \mathbf{Q}_{m(n+1)}$. This is done by introducing a new variable $W \in \mathbf{Q}_{m}$ defined, similarly to \eqref{eq:def_W_opt}, as 
\alg{\label{eq:def_W}	W:= Z_{0}-V_{00}+\lambda  \hat P_0\succ 0 }
such that, as in \eqref{eq:def_V_opt}, the variable $V$  can be expressed as
\alg{\label{eq:def_V} V= \lambda  T(\hat P)+T(Z)- \left[\begin{array}{cc}W & 0  \\0 & 0 \end{array}\right] . }	 
Accordingly, the dual problem \eqref{dual} can be expressed in terms of the variables $\lambda$, $W$ and $Z$ as follows:
\begin{equation}\label{eq:dual_W}
\min_{(\lambda, W, Z) \in \mathcal{C}} J
\end{equation}
where 
\[
J:=  \lambda \Big(- \log \big| \lambda^{-1} W \big| -\int \log|\hat{\Phi}|  + \delta \Big) \]
and the corresponding feasible set $ \mathcal{C}$ is:
\begin{align*}
\mathcal{C} := \lbrace & ( \lambda, W, Z):  W \in \mathbf{Q}_{m},  W \succ 0, Z \in \mathcal{O}, \lambda \in \mathbb{R},  \\& 
\lambda >0,\lambda  T(\hat P) + T(Z)- \left[\begin{array}{cc}W & 0  \\0 & 0 \end{array}\right] \succeq 0, \\& I + \lambda  T(\hat P) - \left[\begin{array}{cc}W & 0  \\0 & 0 \end{array}\right] \succeq 0 \rbrace.
\end{align*}
We can further simplify our problem as follows. First, we observe that the constraint 
\beq \label{eq:V_sdp}
V = \lambda  T(\hat P) + T(Z)- \left[\begin{array}{cc}W & 0  \\0 & 0 \end{array}\right] \succeq 0
\eeq
implies 
$$
\lambda  T(\hat P) + T(Z) \succeq \left[\begin{array}{cc}W & 0  \\0 & 0 \end{array}\right] 
$$
and then, by Lemma \ref{LemmaYW}, $ \lambda  T(\hat P) + T(Z) \succ 0 $. 
Now, we can easily rewrite \eqref{eq:V_sdp} recalling the characterization of a symmetric positive semidefinite matrix using the Schur complement. 
To this aim, it is convenient to introduce the linear operators $ T_{0,0} : \mathbf{M}_{m,n} \to \mathbf{Q}_m,$  $ T_{0,1:n} : \mathbf{M}_{m,n} \to \mathbf{M}_{m,n-1}$ and 
$ T_{1:n,1:n} : \mathbf{M}_{m,n} \to \mathbf{Q}_{mn}$ that, for a given matrix $ H \in \mathbf{M}_{m,n} $ construct a symmetric block-Toeplitz matrix and extract the blocks in position $(0,0)$, $(0,1:n)$ and $(1:n,1:n)$, respectively.
With this notation, we have 
$$
T(Z + \lambda \hat{P} ) = \left[\begin{array}{cc} T_{0,0}(Z + \lambda \hat{P} ) & T_{0,1:n}(Z + \lambda \hat{P} )  \\ 
T_{0,1:n}^\top(Z + \lambda \hat{P} ) & T_{1:n,1:n}(Z + \lambda \hat{P} ) \end{array}\right]
$$
and the constraint \eqref{eq:V_sdp} is equivalent to require $ T_{1:n,1:n}(Z + \lambda \hat{P} )  \succ 0$ {and} $ W  \preceq Q (\lambda, Z)  $
with 
\begin{align*}
	Q (\lambda, Z) :=   T_{0,0}(Z + \lambda \hat{P} )  - T_{0,1:n}(Z + \lambda \hat{P} ) \times T_{1:n,1:n}^{-1} (Z +\\  \lambda \hat{P} )  T_{0,1:n}^\top(Z + \lambda \hat{P} ).
\end{align*} 
%In a similar fashion, by noticing that $I + \lambda T(\hat{P}) \succ 0 $ and computing the Schur complement of the south-east block, the constraint
%\beq \label{eq:Wconstraint_sdp}
%I + \lambda  T(\hat P) - \left[\begin{array}{cc}W & 0  \\0 & 0 \end{array}\right] \succeq 0
%\eeq
%can be equivalently expressed as $W \preceq R(\lambda)$ where
% $$ 
%R(\lambda) :=  I + T_{0,0} ( \lambda \hat{P} ) - T_{0,1:n}(\lambda \hat{P} ) \big(I + T_{1:n,1:n} (\lambda \hat{P})  \big)^{-1} T_{0,1:n}^\top(\lambda \hat{P} ). $$

In a similar fashion, the last matricial inequality constraint in $\mathcal{C}$ can be equivalently expressed as $W \preceq R(\lambda)$ where
\begin{align*}
	R(\lambda) :=  I + T_{0,0} ( \lambda \hat{P} ) - T_{0,1:n}(\lambda \hat{P} ) \big(I + T_{1:n,1:n} (\lambda \hat{P})  \big)^{-1}  \times \\  T_{0,1:n}^\top(\lambda \hat{P} ).
\end{align*}
Therefore, Problem \eqref{dual} can be formulated as 
\begin{equation}\label{eq:dual_W_final}
\min_{(\lambda, W, Z) \in \mathcal{C}} J =  \lambda \Big(- \log \big| \lambda^{-1} W \big| -\int \log|\hat{\Phi}|  + \delta \Big)
\end{equation}
where 
\begin{align*}
\mathcal{C} := & \lbrace (\lambda, W, Z):  Z \in \mathcal{O}, \lambda \in \mathbb{R}, \lambda >0, T_{1:n,1:n}(Z + \lambda \hat{P} )  \succ 0, \\&
 W \in \mathbf{Q}_{m}, W \succ 0,  W  \preceq Q (\lambda, Z), \; W \preceq R(\lambda) \rbrace.
\end{align*}

Solving Problem \eqref{eq:dual_W_final} simultaneously for $\lambda$, $W,$ and $Z$ is not trivial because the inequality constraints $W  \preceq Q (\lambda, Z)$ and $W \preceq R(\lambda)$ both depend on $\lambda$. On the other hand, 
once we fix the dual variable $\lambda$ to a positive constant $\bar{\lambda}>0$, the problem:
\beq \label{pb:lambda_fixed}
\min_{(W, Z) \in \mathcal{C}_{\bar \lambda} } J(\bar{\lambda}, W, Z)
\eeq 
with 
\begin{align*}
\mathcal{C}_{\bar \lambda} := \lbrace & (W, Z):  Z \in \mathcal{O}, \; W \in \mathbf{Q}_{m},\; T_{1:n,1:n}(Z + \bar{\lambda} \hat{P} )  \succ 0, \\ & W \succ 0,  \; W  \preceq Q (\bar{\lambda}, Z), \; W \preceq R(\bar{\lambda}) \rbrace. 
\end{align*}
can be efficiently solved by resorting to the ADMM algorithm \cite{BoydADMM}.
To this aim, we rewrite Problem \eqref{pb:lambda_fixed} by introducing a new variable $ Y \in \mathbf{Q}_m $  defined as $Y = Q(\bar{\lambda},Z) - W:$ 
\begin{equation}
\label{eq:dual_ADMM}
\begin{aligned}
	\min_{\substack{ (W, Z) \in \mathcal{C}_{W,Z}, \\ {Y \in \mathbf{Q}_m^+} }}  & J = \bar{\lambda} \big(- \log \big| \bar{\lambda}^{-1} W \big| -\int \log|\hat{\Phi}|  + \delta \big) \\
	\text{ subject to }  & Y =  Q(\bar{\lambda},Z) - W
\end{aligned}
\end{equation}
where
\begin{align*}
\mathcal{C}_{W, Z}  :=  \lbrace & ( W, Z): Z \in \mathcal{O}, W \in \mathbf{Q}_m, \; W \succ 0, \\ &  W \preceq R(\bar{\lambda}), \; T_{1:n,1:n}(Z + \bar{\lambda} \hat{P} ) \succ 0  \rbrace
\end{align*}
and  $\mathbf{Q}_m^+ $ denotes the cone of symmetric positive semidefinite matrices of size $ m \times m $. The \textit{augmented Lagrangian} for \eqref{eq:dual_ADMM} is:
\begin{align*}  
\mathcal{L}_{\rho}  (W, Z, Y, M) := \bar{\lambda}  \Big(- \log \big| \bar{\lambda}^{-1} W \big| -\int \log|\hat{\Phi}|  + \delta \Big) + 
	\\ \langle M,  Y - Q (\bar{\lambda}, Z) +  W \rangle + \frac{\rho}{2} \p Y - Q (\bar{\lambda}, Z) + W  \p ^2  
\end{align*}
where $ M \in \mathbf{Q}_m $ is the Lagrange multiplier, and  $ \rho > 0$ is the \textit{penalty parameter}. 
Accordingly, given the initial guesses $ W^{(0)},$ $Z^{(0)},$ $Y^{(0)}$ and $M^{(0)}$, the ADMM updates are:
\begin{align}
\label{eq:ADMM_update1}  & ( W^{(k+1)},  Z^{(k+1)} )  = \argmin_{ (W, Z) \in \mathcal{C}_{ W, Z} }  \mathcal{L}_{\rho} (W, Z, Y^{(k)}, M^{(k)} ) \\
\label{eq:ADMM_update2}  & Y^{(k+1)}  = \argmin_{Y  \in \mathbf{Q}_m^+  }  \mathcal{L}_{\rho} (W^{(k+1)}, Z^{(k+1)}, Y, M^{(k)}) \\
& M^{(k+1)} = M^{(k)} + \rho \big(  Y^{(k+1) } - Q (\bar{\lambda}, Z^{(k+1)}) + W^{(k+1)} \big). \nonumber
\end{align}

Problem \eqref{eq:ADMM_update1} does not admit a closed form solution, therefore we approximate the optimal solution by a gradient projection step:
\begin{align*}
W^{(k+1)} & = \Pi \big( W^{(k)} - t_k \nabla_{W}  \mathcal{L}_{\rho} (W^{(k)}, Z^{(k)} , Y^{(k)}, M^{(k)} ) \big) \\
Z^{(k+1)} & = \Pi_{\mathcal{O}} \; \big( Z^{(k)} - t_k \nabla_{Z}  \mathcal{L}_{\rho} (W^{(k)} , Z^{(k)} , Y^{(k)}, M^{(k)}) \big)
\end{align*}
where: 

\begin{itemize}
\item  $\nabla_{W}  \mathcal{L}_{\rho} (W, Z , Y, M )$ denotes the gradient of the augmented Lagrangian with respect to $W$:
$$ \nabla_{W}  \mathcal{L}_{\rho} = -\bar{\lambda} W^{-1} + M + \rho(Y - Q + W). $$

\item  $\nabla_{Z}  \mathcal{L}_{\rho} (W,  Z , Y, M )$ denotes the gradient of the augmented Lagrangian with respect to $Z$:
\begin{align*} 
	\nabla_{Z}  \mathcal{L}_{\rho} = D \big( \left[ \begin{array}{c} I_m \\ -T^{-1}_{1:n,1:n} T^\top_{0,1:n}  \end{array}\right]  \big(- M  - \rho  (Y -\\ Q  + 
	  W ) \big) \left[ \begin{array}{cc} I_m & -T_{0,1:n} T^{-1}_{1:n,1:n} \end{array} \right] \big)
\end{align*}
where the omitted argument of the operators $ T_{0,1:n} $ and $ T_{1:n,1:n}$ is intended to be equal to $(Z + \bar{\lambda} \hat{P} ).$

\item $\Pi_{\mathcal{O}}$ denotes the projection operator onto $\mathcal{O}$:
$$ \Pi_{\mathcal{O}} (A) = \ofd_B (A).$$	

\item $\Pi$ denotes the projection operator onto the convex cone $ \{S \in \mathbf{Q}_m :  S \preceq R(\bar{\lambda} ) \}. $ 
%	defined as:
%	\beq \label{eq:projectors_Pi}
%	\Pi (A) := \argmin_{S \in \mathbf{Q}_m, S \preceq R(\bar{\lambda}) } \p S - A \p^2 .
%	\eeq
It is not difficult to see that
$$ 	\Pi (A) =R(\bar{\lambda}) - \Pi_+ ( R(\bar{\lambda}) - A ), $$
where $ \Pi_+ $ is the projection operator onto the cone $\mathbf{Q}_m^+ $.
%	\beq \label{eq:def_projector_alg4}
%	\Pi_{\mathbf{Q}_m^+} (A) := \argmin_{S \in \mathbf{Q}_m^+ } \quad \p S - A \p^2.
%	\eeq 
\item the step-size $t_k$ is determined at each step $k$ in an iterative fashion: we start by setting $t_k =1$ and we decrease it
progressively of a factor $\beta,$ with $ 0  < \beta < 1,$  until the conditions $ W^{(k+1)} \succ 0 $ and $ T_{1:n,1:n}(Z^{(k+1)} + \bar{\lambda} \hat{P} ) \succ 0 $ are met and the Armijo's condition \cite{boyd:vandenberghe:2004} is satisfied.
\end{itemize}
Problem \eqref{eq:ADMM_update2} admits a closed form solution, which can be easily computed as: 
$$	Y^{(k+1)} =   \Pi_+  \Big(Q (\bar{\lambda}, Z^{(k+1)}) -W^{(k+1)}  - \frac{1}{\rho} M^{(k)}  \Big). $$ 
To define the stopping criterion, we need to introduce the following quantities
\begin{align*}
R^P &= Y - Q (\bar{\lambda}, Z^{(k+1)}) + W^{(k+1)} \\
R^D &=  \begin{aligned}[t] D \big(  \left[ \begin{array}{c} I_m \\ -T^{-1}_{1:n,1:n} T^\top_{0,1:n}  \end{array}\right]  (\rho (Y^{(k +1 )} - Y^{(k)}) ) \times \\  \left[ \begin{array}{cc} I_m & -T_{0,1:n} T^{-1}_{1:n,1:n} \end{array} \right] \big) \end{aligned}
\end{align*}
which are referred to as the primal and dual residual, respectively. Notice that the omitted argument of the operators $ T_{0,1:n} $ and $ T_{1:n,1:n}$ is intended to be equal to $(Z^{(k+1)}+ \bar{\lambda} \hat{P} )$. \\
Then, the algorithm stops when the following conditions are met:
\begin{align*}
	\Vert R^P \Vert &\leq \begin{aligned}[t] m \varepsilon^{\text{ABS}} + \varepsilon^{\text{REL}} \max & \{ \Vert W^{(k)}\Vert,  \Vert Q(\bar \lambda, Z^{(k)})\Vert, \Vert Y^{(k)} \Vert \} \end{aligned} \\
	\Vert R^D \Vert &\leq \begin{aligned}[t] m \sqrt{(n+1)} \varepsilon^{\text{ABS}} +  \varepsilon^{\text{REL}} \Vert  D \big(  \left[ \begin{array}{c} I_m \\ -T^{-1}_{1:n,1:n} T^\top_{0,1:n}  \end{array}\right]  \times  \\
		 M^{(k)} \left[ \begin{array}{cc} I_m & -T_{0,1:n} T^{-1}_{1:n,1:n} \end{array} \right] \big)  \Vert \end{aligned}
\end{align*}
where $\varepsilon^{\text{ABS}}$ and $\varepsilon^{\text{REL}}$ are the desired absolute and relative tolerances.

It remains to determine the optimal value $\lambda^\circ$  for $\lambda$ which solves Problem \eqref{eq:dual_W_final}. To this aim, we exploit the following result (see \cite[pp.87-88]{boyd:vandenberghe:2004}):
\begin{propo}\label{propo:convex_minimization}
If $f$ is convex in $(x,y)$ and $\mathcal{C}$ is a convex non-empty set, then the function
\beq
g(x) = \inf_{y \in \mathcal{C}} f(x,y)
\eeq
is convex in $x$, provided that $g(x) > -\infty$ for some $x$. 
The domain of $g$ is the projection of $\dom (f)$ on its $x$-coordinates.
\end{propo}

This result guarantees that the function
$$
g(\lambda) = \min_{(W, Z) \in \mathcal{C}_{\lambda}} J(\lambda, W, Z)
$$
is convex in $\lambda$. Hence, in order to determine  
$ \lambda^\circ = \argmin_{\lambda > 0} g(\lambda) $
we can choose an initial interval of uncertainty $[a,b]$ containing $\lambda^\circ$, and we progressively reduce it by evaluating $g(\lambda)$ at two points within the interval placed symmetrically, each at distance $h > 0 $ from the midpoint. This is repeated until the width of the uncertainty interval is smaller than a certain tolerance $l > 0$.

The overall procedure to solve the dual problem \eqref{eq:dual_W_final} is summarized in Algorithm~\ref{algo:RDFA}.

\begin{algorithm} 	\caption{}	
\textbf{Input:} $b > a > 0$, $ l > 0$,  $ h > 0$ \\
\textbf{Output:} $(\lambda^\circ,W^\circ, Z^\circ)$
\begin{algorithmic}[1]
	\REPEAT
	\STATE $\tilde{a} = (a+b)/2 - h; \; \tilde{b} = (a+b)/2 + h.$
	\STATE Compute $g(\tilde{a})$ by applying the ADMM with $\lambda =\tilde{a}.$ 
	\STATE Compute $g(\tilde{b})$ by applying the ADMM with $\lambda =\tilde{b}.$ 
	\IF {$g(\tilde{a}) < g(\tilde{b}) $}
	\STATE $b = \tilde{b}$
	\ELSE 
	\STATE $a = \tilde{a}$
	\ENDIF
	\UNTIL{$b-a < l$}	
	\STATE $\lambda^\circ = (a+b)/2$.
	\STATE Compute $(W^\circ, Z^\circ)$ by applying the ADMM with $\lambda = \lambda^\circ.$
\end{algorithmic} 
\label{algo:RDFA}
\end{algorithm}

%%%%%%%%%%%%%%%%%%%%%%%%%%%%%%%%%%%%%%%%%
\section{Identification of ARMA factor models}\label{Sec:ARMAid}
In this section we extend the proposed approach to { ARMA processes.} Consider the ARMA factor model:
\beq \label{ARMA_FA} y(t)= a^{-1}(W_L u(t)+ W_D w(t) ) \eeq
where
$$ a(e^{i\vartheta })=\sum_{k=0}^p a_{k}e^{-i\vartheta k}, \quad a_{k}\in \Rs $$ 
and $W_L, W_D$, $u$ and $w$ are defined analogously to \eqref{MA_FA}.
%$u=\{u(t),\; t\in\Zs\}$ and $w=\{w(t),\; t\in\Zs\}$ are i.i.d. Gaussian processes of dimension $m$ and $r$, respectively, with covariance equal to the identity.
Notice that $y_{MA} (t) := a y(t) = W_L u(t)+ W_D w(t) $
is a MA process of order $n$ whose spectral density $ \Phi= W_L W_L^*+ W_D W_D^*\in\Qc_{m,n}$ admits a low rank plus diagonal decomposition. Finally, it is worth noting that it is not restrictive to assume that the autoregressive part in (\ref{ARMA_FA}) is characterized by a scalar filter $a$; Indeed, any ARMA factor model can be written in the form of (\ref{ARMA_FA}).

Assume now to collect a realization $\mathrm y^N=\{\,\mathrm y(1) \ldots \mathrm y(N) \,\}$  of numerosity $N$ of the process $y$.
Our aim is to estimate the factor model \eqref{ARMA_FA} and the number of factors $r$. Before proceeding, the following observation needs to be made: there is an identifiability issue in the problem. Indeed, if we multiply $a(z)$, $W_L$ and $W_D$ by an arbitrary non-zero real number $c$, the model remains the same. We can easily eliminate this uninteresting degree of freedom by normalizing the polynomial $a(z)$, so that from now on we assume $a_0 = 1$.

The idea is to estimate first $a$, and then $\Phi_L$ and $\Phi_D$ by preprocessing $\mathrm y^N$ through $a$.
In more detail, the proposed solution consists of the following two steps:
\begin{enumerate}
\item  The \emph{AR dynamic estimation.}  Given the realization $\mathrm y^N$, we estimate the $p$ parameters of the filter $a$ by applying the maximum likelihood estimator proposed in \cite[Section II.b]{CDC2020_Crescente}.
In doing so, we are estimating an AR process whose spectral density is $a^{-1} (a^{-1})^* I_m$.  
\item The \emph{MA dynamic factor analysis.}
Let $\mathrm y_{MA}^N$ be the finite length trajectory obtained by passing through the filter  $a^\circ (e^{i\vartheta})$ the trajectory $\mathrm y^N$ with zero initial conditions. After computing the truncated periodogram $\hat{\Phi} \in\Qc_{m,n} $ from $\mathrm y_{MA} ^N $, we solve Problem \eqref{new} with $\hat{\Phi}$ in order to recover the number of latent factors.
\end{enumerate}
Although the above procedure is suboptimal, the numerical simulations showed that the resulting estimator of the number of factors performs well, see Section \ref{sec:SIM_ARMA}.

%%%%%%%%%%%%%%%%%%% 
\section{Numerical simulations}\label{sec:sim}
In this section, we test the performance of the proposed approach both for MA and ARMA factor models. 
In all the simulations, the parameter $\delta$ is computed according to the empirical procedure of Section \ref{Sec:delta} for $\alpha = 0.5$. Then, Problem \eqref{eq:dual_W_final} is solved by applying Algorithm \ref{algo:RDFA} with $l = 7 $ and $  h = 3 .$ In regard to the ADMM algorithm, we set $\varepsilon^{\text{ABS}} = 10^{-4}$, $\varepsilon^{\text{REL}} = 10^{-4}$ and the penalty term $ \rho = 0.05.$ 

\subsection{Synthetic Example - MA factor models}
\begin{figure}
	\centering
	\includegraphics[width=\linewidth]{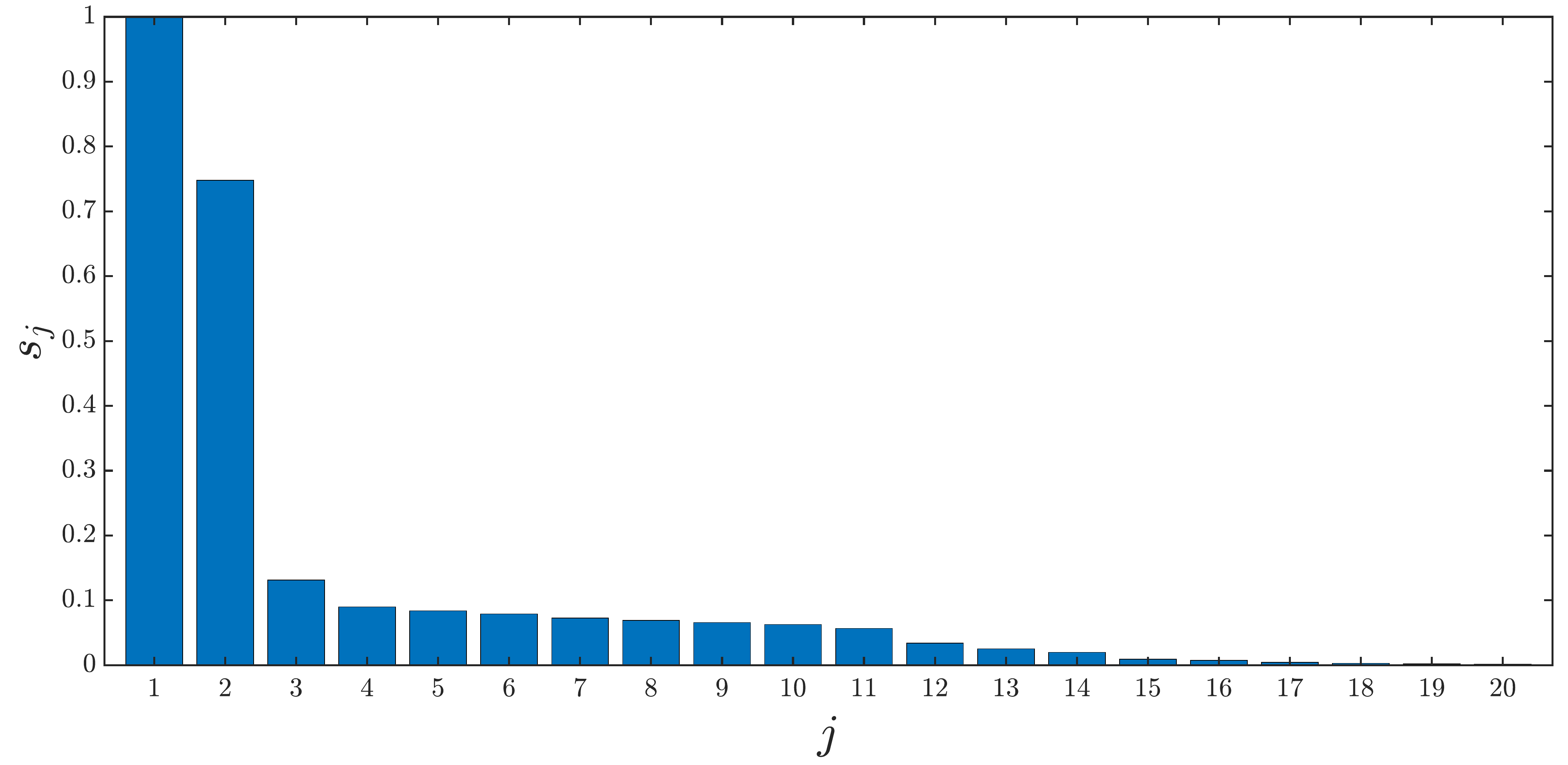}
\caption{Estimated MA factor model with $n=2$, $m=40$, and $r=2$. Integral over the unit circle of the first 20 normalized singular values of $\Phi_L^\circ$ with $N = 5000$.} 
	\label{fig:factor_r2}
\end{figure}
\begin{figure}
\centering
\includegraphics[width=\linewidth]{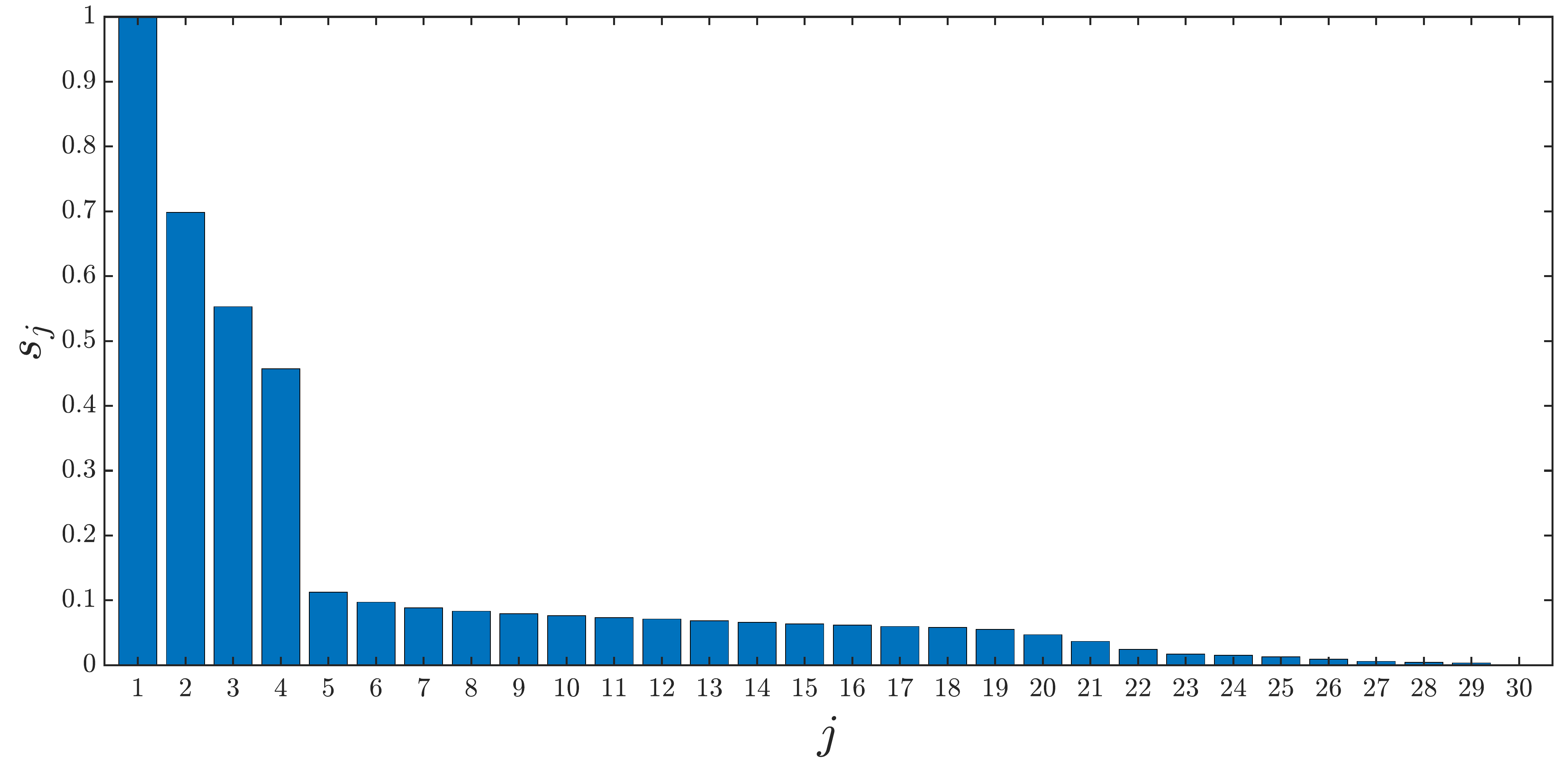}	
\caption{Estimated MA factor model with $n=2$, $m=40$, and $r=4$. Integral over the unit circle of the first 30 normalized singular values of $\Phi_L^\circ$ with $N = 5000$}
\label{fig:factor_r4}
\end{figure}

We consider an MA factor model \eqref{MA_FA} of order $n=2$, with $m=40$ manifest variable and $r=2$ latent factors, computed by randomly generating the zeros of the transfer functions $[W_{L}]_{(i,j)}$'s and $[W_D]_{(i,i)}$'s for $i=1,\dots, m,$ $j=1,\dots,r$ within the circle with center at the origin and radius $0.95$ on the complex plane.   
It is worth noting that $\int \Vert \Phi_L(e^{i\theta}) ) \Vert = 141.83 $ and $\int \Vert \Phi_D(e^{i\theta}) ) \Vert = 31.29 $, that is the idiosyncratic component is not negligible with respect to the latent variable. 
We generate from the model a sample $\mathrm y^N$ of length $N=5000$ and we apply the proposed identification procedure to estimate the number of common factors. 
We define 
$$  
s_j  := \int \frac{\sigma_j(\Phi_L^\circ (e^{i\theta}) ) }{\sigma_1(\Phi_L^\circ (e^{i\theta}) )}
$$
where $ \sigma_j(\Phi_L^\circ (e^{i\theta}) $ denotes the $j-th$ largest eigenvalue of $ \Phi_L^\circ $ at frequency $\theta$. It is clear that $s_j$ represents the integral of the $j-th$ largest normalized singular value of $ \Phi_L^\circ $ over the unit circle. 
The quantities $ s_j$ are plotted in Figure \ref{fig:factor_r2}; we can notice that there is a knee point at $j=2$, so that the numerical rank of $ \Phi_L^\circ$ is equal to $2$ and, in doing so, we can recover the exact number of common factors.

The effectiveness of the proposed estimator is also tested by considering a data sample of numerosity $N=5000$ generated by a MA factor model with $n=2$, $m=40$ and $r=4$ common factors; the integrals of the Frobenius norm of $\Phi_L$  and $\Phi_D$  are $80.76$ and $8.89,$ respectively. 
As showed in Figure \ref{fig:factor_r4}, even in this case we are able to estimate the correct number of latent variables.

Finally, we obtained similar results with different samples and by changing the ``true'' factor model from which we generated the data. 

\subsection{Synthetic Example - ARMA factor models}\label{sec:SIM_ARMA}
\begin{figure}
	\centering
	\includegraphics[width=\linewidth]{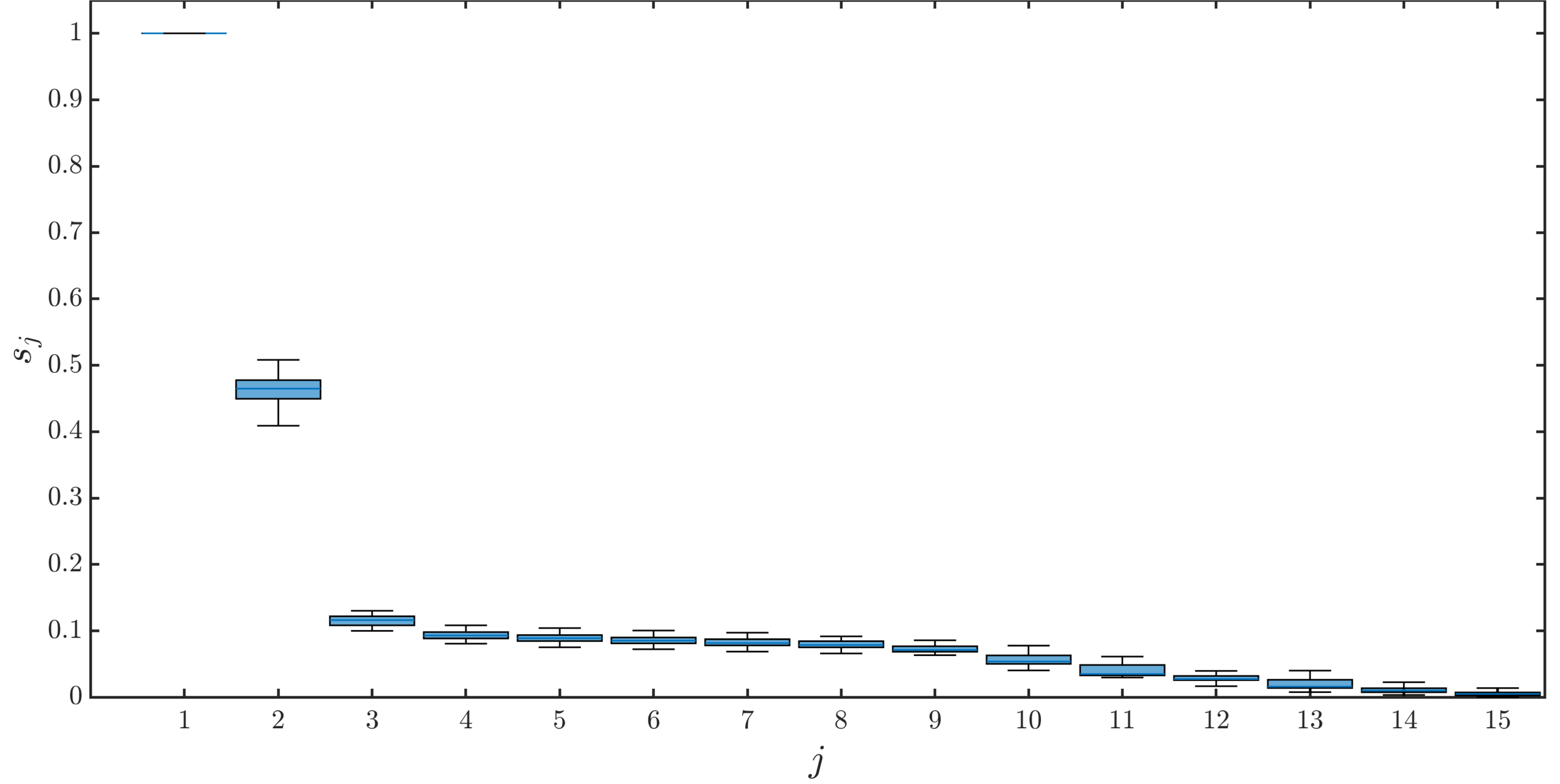}
	\caption{ Estimated ARMA factor model with $m=40,$ $r=2,$ $n=2$ and $p=2.$ Box-plot of the integral over the unit circle of the first 15 normalized singular values of $\Phi_L^\circ$ with $N = 5000$. }
	\label{fig:ARMA_boxplot}
\end{figure}
To provide empirical evidence of the estimation performance of the algorithm of Section \ref{Sec:ARMAid}, a Monte Carlo simulation study composed of 50 experiments is performed.
We randomly build an ARMA factor model \eqref{ARMA_FA} with $m=40$, $r=2$, $n=2$ and $p=2$; without loss of generality we fix $a_0 =1$.  Then, for each Monte Carlo experiment a data sequence of length $N = 5000$ is randomly generated from the model and the ARMA factor model identification procedure is performed.
The boxplot of the quantities $s_j$ for the estimated $\Phi_L^\circ$'s are shown in Figure \ref{fig:ARMA_boxplot} and it reveals that the proposed identification procedure is able to successfully recover the number of latent factors.

\subsection{Smart Building Dataset}

The SMLsystem is a house built in Valencia at the Universidad CEU Cardenal Herrera (CEU-UCH).
%to participate in the Solar Decathlon 2012 competition
%,an international competition among Universities that
%promotes research in the development of energy-efficient houses.  
It is a modular house that integrates a whole range
of different technologies to improve energy efficiency, with the objective to construct a near zero-energy house. 
A complex monitoring system has been used in the SMLsystem: it has indoor sensors for
temperature, humidity and carbon dioxide; outdoor sensors are also available
for lighting measurements, wind speed, rain, sun irradiance and
temperature. We refer the reader to \cite{ZAMORAMARTINEZ2014} for a detailed description of the building and its monitoring system.
Two datasets from the SMLsystem are available for download at the UCI Machine
Learning repository \url{http://archive.ics.uci.edu/ml}. We take into account $m=17$ sensor signals extracted from these datasets: the indoor temperature (in $^\circ \text{C} $) of the dinning-room and of the room, the weather forecast temperature (in $^\circ \text{C} $), the carbon dioxide (in ppm) in the dinning room and in the room, the relative humidity (in \%) in the dinning room and the room, the lighting in the dinning room and the room (in lx), the sun dusk, the wind (in cm/sec), the sun light (in klx) in the west, east and south facade, the sun irradiance (in dW), the outdoor temperature (in $^\circ \text{C} $) and finally the outdoor relative humidity (in \%).
The data are sampled with a
period of $T = 15 \text{min}$ and each sample is the mean of the last quarter,
reducing in this way the signal noise. The first dataset $\mathrm{y}^{N_1} = \{  \, \mathrm y(1), \ldots, \mathrm y(N_1) \, \}$ was captured during 
March 2011 and has $N_1 = 2764$ points ($ \approx 28$ days), while the second dataset $\mathrm{y}^{N_2} = \{ \, \mathrm y(N_1 + 1), \ldots, \mathrm y(N_1 + N_2) \, \}$ has $N_2 = 1373$ points ($ \approx 14$ days) collected in June 2011.

It is reasonable to expect that the variability of the considered signals may be successfully explained by a smaller number of factors.
Motivated by this reason, we apply the ARMA factor model identification procedure with parameters $n=2$ and $p=2$ using the realization $\mathrm{y}^{N_1}$. As shown in Figure \ref{fig:bar_domhouse}, we obtain an estimate of $ r = 4 $ latent factors.

For the sake of comparison, we also use the Matlab function \texttt{armax()} of the System Identification Toolbox to compute the prediction-error  method (PEM) estimate  for an ARMA model with polynomials $A(z)$ and $B(z)$ of order 2 and $A(z)$ diagonal from the realization $\mathrm{y}^{N_1}$.

Finally, the second dataset $\mathrm{y}^{N_2}$ is used in the validation step to test the prediction capability of the two estimated ARMA models.
The results are summarized in Figure \ref{fig:fit_domhouse} which displays for each output channel $j=1,\dots,m$ the fit (percentage) term:
$$
J_{FIT, j} := 100 \left( 1 - \frac{\sqrt{\sum_{t=N_1 + 1}^{N_1 + N_2} (\mathrm y_j(t) - \hat{\mathrm y}_j(t|t-1))^2 }}{ \sqrt{\sum_{t=N_1 + 1}^{N_1 + N_2} (\mathrm y_j(t) - \bar{\mathrm y}_j)^2 }} \right) 
$$
where $\bar{\mathrm y}_j := \frac{1}{N_2} \sum_{t=N_1 + 1}^{N_1 + N_2} \mathrm y_j(t)$ and $\hat{\mathrm y}_j(t|t-1)$ is the one-step ahead prediction at time $t$ computed with zero initial conditions.
The figure shows that the ARMA factor model matches quite well the measurement data $\mathrm{y}^{N_2}$, reaching fit values similar to the PEM estimate. This is a remarkable result as the performances of the two approaches are essentially the same, but the factor model is parameterized by 257 coefficients, much less than the 869 coefficients of the PEM estimate.

\begin{figure}
	\centering
	\includegraphics[width=\linewidth]{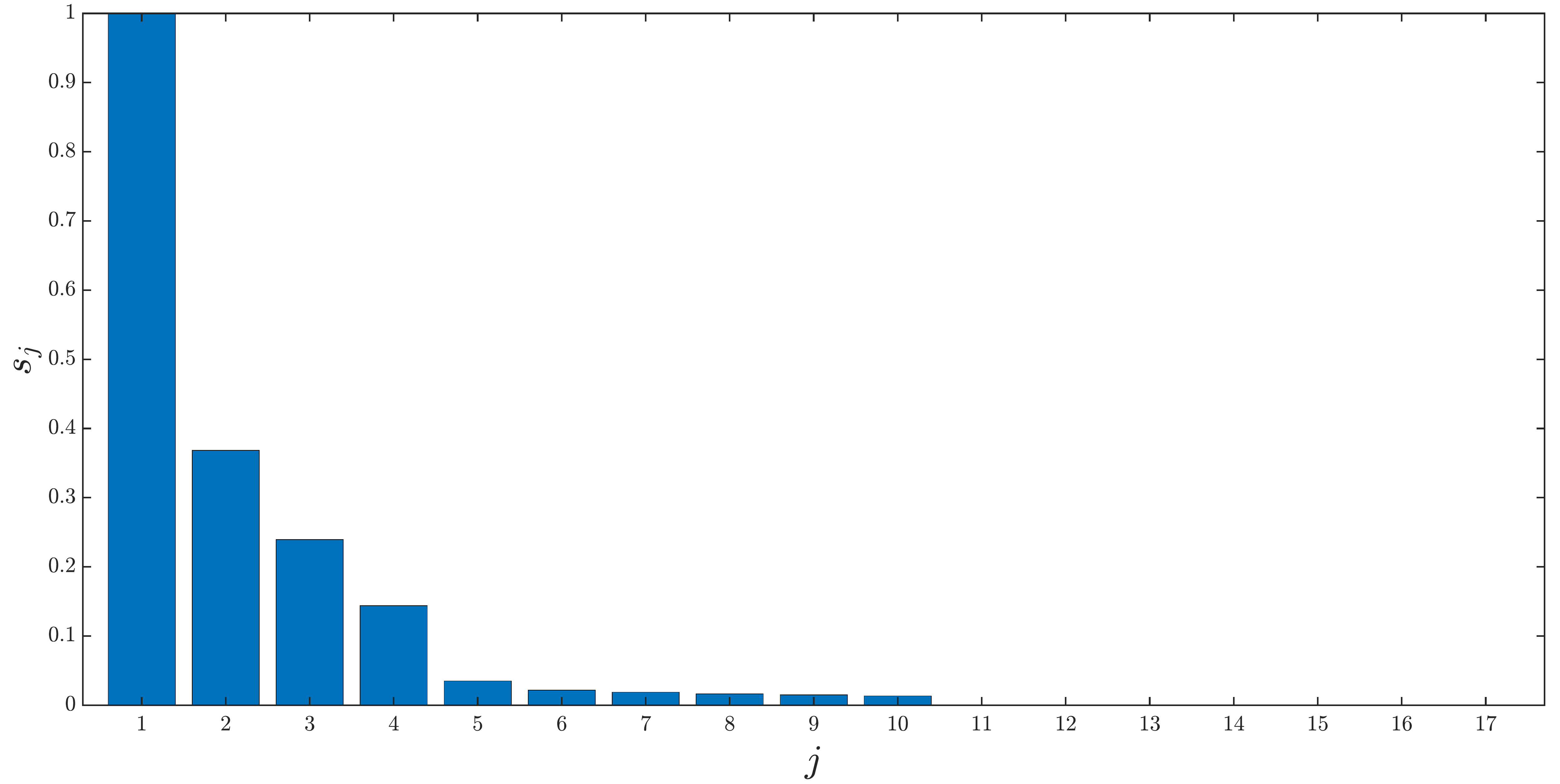}
	\caption{Application of the ARMA factor models identification procedure by using the measurements $\mathrm{y}^{N_1}$ from the SMLsystem as training data. The figure shows the integral over the unit circle of the normalized singular values of $\Phi_L^o$. }
	\label{fig:bar_domhouse}
\end{figure}
\begin{figure}
	\centering
	\includegraphics[width=\linewidth]{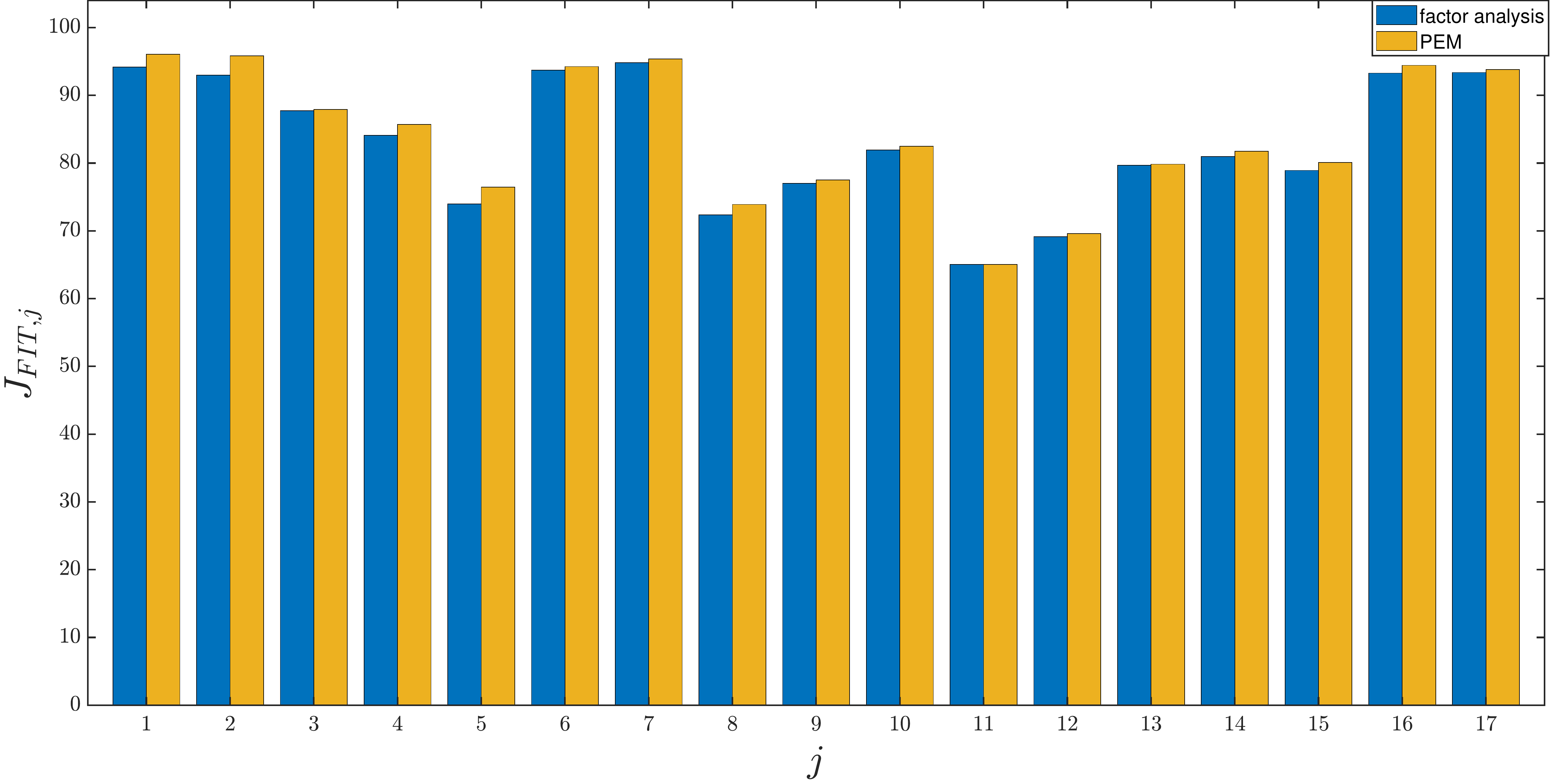}
	\caption{Fit (in percentage) term $J_{FIT,j}$ for each output channel for the model estimated via factor analysis and via PEM. The fit values are computed by using the measurements $\mathrm{y}^{N_2}$ from the SMLsystem as validation data.}
	\label{fig:fit_domhouse}
\end{figure}

\section{Conclusion}\label{sec:concl}
A procedure to estimate the number of factors and to learn ARMA factor models
has been proposed.  This method is based on the solution of an optimization problem 
whose solution has been proven to exist via dual analysis.
The simulations results applying the procedure both to synthetic and real data provide
evidence of a good performance.

\appendix

\subsection*{Proof of Lemma \ref{lemm:Jensen-Kolmogorov}} \label{app:JK}
Since $\Phi=\Delta X \Delta^*$ with $X\succeq 0$, there exists $A\in \mathbb{R}^{m\times m(n+1)}$ such that $X=A^\top A$. 
The matrix $A$ is such that $\Phi \succeq 0$ admits the spectral factorization $ \Phi = W W^*$ where $W := \Delta A^\top$. 
Now, define $\Phi_n := \Phi + \frac{1}{n} I $ with $ n \in \mathbb{N} $ and let $ W_{n} := \Delta A_n$  be a spectral factor 
of $\Phi_n$ with $A_n \in \mathbb{R}^{m\times m(n+1)}.$ 	Clearly, $\lim_{n \to +\infty} \Phi_n = \Phi $; accordingly, $\lim_{n \to +\infty} W_{n} = W$ and $\lim_{n \to +\infty} A_{n} = A$. 	
Since $\Phi_{n}\succ 0$ $\forall \vartheta$ we can exploit \eqref{eq::JKformula} to obtain 
$$ \int  \log|\Phi_n| = \log|A_{n_0}^\top A_{n_0}|. $$
Then, applying the limit operator to both sides, we have
\begin{align*}
	\lim_{n \to +\infty}  \int  \log|\Phi_n| % = \lim_{n \to +\infty} \log |A_{n_0}^\top A_{n_0}| 
	= \log|A_0^\top A_0| = \log |X_{00}|.
\end{align*} 	
To conclude the proof, it remains to show that in the left side of the previous equation it is possible to 
interchange the limit and the integral operators.
%	so that
%	$$
%	\lim_{n \to +\infty}  \int  \log|\Phi_n| =   \int \lim_{n \to +\infty} \log|\Phi_n| = \int \log |\Phi|.
%	$$
To this aim, we introduce the sequence $ \{ f_n \}_{n=1}^{+ \infty} $ where $ f_n(t):= \log|\Phi_n (\vartheta)|$ 
%	$$ f_n(t):= \log|\Phi_n (\vartheta)| = \log|\Phi(\vartheta) + \frac{1}{n} I|$$
and the function  $f(\vartheta):=\lim_{n\rightarrow + \infty}f_n(t) = \log|\Phi (\vartheta)| $. 
Observe that, since the interval of integration  $[-\pi, \pi]$ is bounded and $ f_1(\vartheta)  < +\infty$ for any
$\vartheta\in [-\pi, \pi]$, then $ \int f_1(\vartheta) d\vartheta < +\infty. $	
We also define the sequence  $ \{ g_n \}_{n=1}^{+ \infty} $  as $g_n(\vartheta) := f_n(\vartheta)-f_1(\vartheta)$ and $g(\vartheta):=\lim_{n\rightarrow + \infty}g_n(\vartheta).$ 
$ \{ g_n \}$ is a pointwise non-increasing  sequence of measurable non-positive functions, 
$$   \dots \leq g_2 (\vartheta) \leq g_1 (\vartheta) \leq 0 , \quad \forall \vartheta \in \left[  -\pi, + \pi \right] $$
converging to $g(\vartheta)$ from above. Hence, it satisfies all the hypotheses of Beppo-Levi's monotone convergence theorem (applied with opposite signs), from which it immediately follows that
$$
\lim_{n\rightarrow + \infty} \int g_n(\vartheta) =  \int g(\vartheta),
$$
and consequently 
{ \beq \label{eq:integral_fn} \lim_{n\rightarrow + \infty} \int f_n(\vartheta) =\int g(\vartheta) + \int f_1(\vartheta) . \eeq }	
%	\begin{align*}
%	\lim_{n\rightarrow + \infty} \int f_n(\vartheta)
%	% & = \lim_{n\rightarrow + \infty} \int f_n(\vartheta) - f_1( \vartheta)  + \int f_1(\vartheta) \\
%	% & = \lim_{n\rightarrow + \infty} \int g_n(\vartheta) + \int f_1(\vartheta) \\
%	& =\int g(\vartheta) dt + \int f_1(\vartheta) . \numberthis \label{eq:integral_fn}
%	\end{align*}
Now, since $f_1(\vartheta) < +\infty$ for all $ \vartheta $, 
\beq \label{eq:limit_gn}
g(\vartheta) = f(\vartheta)-f_1(\vartheta),
\eeq 
and, by plugging \eqref{eq:limit_gn} into \eqref{eq:integral_fn}, we finally obtain
$$
\lim_{n\rightarrow + \infty} \int f_n(\vartheta) =\int f(\vartheta).	\quad  \blacksquare
$$

\subsection*{Proof of Proposition \ref{prop::third_restriction}}
Consider a sequence $ ( \lambda^{(k)} , V^{(k)} , Z^{(k)} )_{k \in \mathbb{N} }$ in $ \mathcal{C}_2 $.

We first show that $ [Z^{(k)}]_0$  cannot diverge. Indeed, assume by contradiction that $\lim_{k \to \infty} \p [Z^{(k)}]_0  \p = + \infty$. 
Since it is a symmetric and traceless matrix, this implies 
\begin{equation}\label{eig_Z0}
	\lim_{k \to \infty} \quad \min _ { \alpha^{(k)} \in \sigma  \big(  [Z^{(k)}]_0  \big) }   \alpha^{(k)}   = - \infty .
\end{equation}
In view of \eqref{eig_Z0}, since $\lambda^{(k)} \hat{P}_0$ is bounded and $V^{(k)} $ positive semidefinite  $\forall k$,  then $(  \lambda^{(k)} \hat{P}_0 + [Z^{(k)}]_0 - [V^{(k)}]_{00} ) $ has at least a negative eigenvalue for $k$ sufficiently large, 
so that the sequence $(\lambda^{(k)} , V^{(k)} , Z^{(k)} )$ is not in $ \mathcal{C}_2$. 
We conclude that 
$$
	\lim_{k \to \infty}\p [Z^{(k)}]_0   \p < \infty .
$$
As a consequence, since $ \beta I \preceq T(Z^{(k)}) - V^{(k)} \preceq I $ (which is one of the condition for the sequence to be in $ \mathcal{C}_2 $ ), 
and $ [ T(Z^{(k)}) ]_{hh} = [Z^{(k)}]_0  $ by construction, 
it holds that $\forall k$
$$
	\p [V^{(k)}]_{hh} \p < \infty, \quad h=0,\dots,n .
$$
Then, from $V^{(k)} \succeq 0$ it follows that also the off-diagonal blocks of $V^{(k)}$ must be bounded $\forall k$, i.e.
\begin{equation} \label{eq::Vlh_bound}
	\p [V^{(k)}]_{hl} \p < \infty, \quad l \neq h, \; \; l, h=0, \dotsc, n.
\end{equation} 
Finally, by the boundedness of $ (T(Z^{(k)}) - V^{(k)} )$  and by \eqref{eq::Vlh_bound} we obtain that $\forall k$
\begin{equation}
	\p [Z^{(k)}]_{h} \p < \infty \quad h=1, \dots ,n ,
\end{equation}
which concludes the proof.  $ \; \blacksquare $

%% REFERENCES
%\bibliographystyle{IEEEtran}
%\bibliography{biblio_RDFA}

\begin{thebibliography}{10}
\providecommand{\url}[1]{#1}
\csname url@samestyle\endcsname
\providecommand{\newblock}{\relax}
\providecommand{\bibinfo}[2]{#2}
\providecommand{\BIBentrySTDinterwordspacing}{\spaceskip=0pt\relax}
\providecommand{\BIBentryALTinterwordstretchfactor}{4}
\providecommand{\BIBentryALTinterwordspacing}{\spaceskip=\fontdimen2\font plus
\BIBentryALTinterwordstretchfactor\fontdimen3\font minus
  \fontdimen4\font\relax}
\providecommand{\BIBforeignlanguage}[2]{{%
\expandafter\ifx\csname l@#1\endcsname\relax
\typeout{** WARNING: IEEEtran.bst: No hyphenation pattern has been}%
\typeout{** loaded for the language `#1'. Using the pattern for}%
\typeout{** the default language instead.}%
\else
\language=\csname l@#1\endcsname
\fi
#2}}
\providecommand{\BIBdecl}{\relax}
\BIBdecl

\bibitem{ning2015linear}
L.~Ning, T.~T. Georgiou, A.~Tannenbaum, and S.~P. Boyd, ``Linear models based
  on noisy data and the {F}risch scheme,'' \emph{SIAM Review}, vol.~57, no.~2,
  pp. 167--197, 2015.

\bibitem{bertsimas2017certifiably}
D.~Bertsimas, M.~S. Copenhaver, and R.~Mazumder, ``Certifiably optimal low rank
  factor analysis,'' \emph{Journal of Machine Learning Research}, vol.~18,
  no.~29, pp. 1--53, 2017.

\bibitem{CFZ_Kybernetika}
V.~{Ciccone}, A.~{Ferrante}, and M.~{Zorzi}, ``Learning latent variable dynamic
  graphical models by confidence sets selection,'' \emph{Kybernetika}, vol.~55,
  no.~4, pp. 74--754, 2019.

\bibitem{zorzi2015factor}
M.~Zorzi and R.~Sepulchre, ``Factor analysis of moving average processes,'' in
  \emph{European Control Conference (ECC)}, Linz, 2015, pp. 3579--3584.

\bibitem{ciccone2017factor}
V.~{Ciccone}, A.~{Ferrante}, and M.~{Zorzi}, ``Factor models with real data: A
  robust estimation of the number of factors,'' \emph{IEEE Transactions on
  Automatic Control}, vol.~64, no.~6, pp. 2412--2425, June 2019.

\bibitem{CFZ_TAC19}
------, ``Learning latent variable dynamic graphical models by confidence sets
  selection,'' \emph{IEEE Transactions on Automatic Control}, vol.~65, no.~12,
  pp. 5130--5143, 2020.

\bibitem{Geweke-77}
J.~Geweke, ``The dynamic factor analysis of economic time series,'' in
  \emph{Latent variables in socio-economic models}, D.~Aigner and
  A.~Goldberger, Eds.\hskip 1em plus 0.5em minus 0.4em\relax Amsterdam:
  North-Holland, 1977.

\bibitem{deistler2007}
M.~Deistler and C.~Zinner, ``Modelling high-dimensional time series by
  generalized linear dynamic factor models: An introductory survey,''
  \emph{Communications in Information \& Systems}, vol.~7, no.~2, pp. 153--166,
  2007.

\bibitem{Stock-Watson-2010}
\BIBentryALTinterwordspacing
J.~Stock and M.~Watson, ``Dynamic factor models,'' 2010, internal report.
  [Online]. Available:
  \url{https://www.princeton.edu/~mwatson/papers/dfm_oup_4.pdf}
\BIBentrySTDinterwordspacing

\bibitem{Figa-Talamanca-et-al}
G.~Fig\'a-Talamanca, S.~Focardi, and M.~Patacca, ``Common dynamic factors for
  cryptocurrencies and multiple pair-trading statistical arbitrages,''
  \emph{Decisions in Economics and Finance}, 2021.

\bibitem{Bottegal-P-15}
G.~{Bottegal} and G.~{Picci}, ``Modeling complex systems by generalized factor
  analysis,'' \emph{IEEE Transactions on Automatic Control}, vol.~60, no.~3,
  pp. 759--774, 2015.

\bibitem{songsiri2010topology}
J.~Songsiri and L.~Vandenberghe, ``Topology selection in graphical models of
  autoregressive processes,'' \emph{Journal of Machine Learning Research},
  vol.~11, no. Oct, pp. 2671--2705, 2010.

\bibitem{ARMA_GRAPH_AVVENTI}
E.~Avventi, A.~Lindquist, and B.~Wahlberg, ``{ARMA} identification of graphical
  models,'' \emph{IEEE Trans. Autom. Control}, vol.~58, no.~5, pp. 1167--1178,
  May 2013.

\bibitem{maanan2017conditional}
S.~Maanan, B.~Dumitrescu, and C.~Giurc{\u{a}}neanu, ``Conditional independence
  graphs for multivariate autoregressive models by convex optimization:
  Efficient algorithms,'' \emph{Signal Processing}, vol. 133, pp. 122--134,
  2017.

\bibitem{Alpago_ELETTERS}
D.~Alpago, M.~Zorzi, and A.~Ferrante, ``Identification of sparse reciprocal
  graphical models,'' \emph{IEEE Control Systems Letters}, vol.~2, no.~4, pp.
  659--664, Oct 2018.

\bibitem{Alpago_TAC}
------, ``A scalable strategy for the identification of latent-variable
  graphical models,'' \emph{Submitted}, 2018.

\bibitem{ZorzSep}
M.~Zorzi and R.~Sepulchre, ``{AR} identification of latent-variable graphical
  models,'' \emph{IEEE Transactions on Automatic Control}, vol.~61, no.~9, pp.
  2327--2340, Sept 2016.

\bibitem{MEManSaid}
S.~Maanan, B.~Dumitrescu, and C.~Giurc{\u{a}}neanu, ``Maximum entropy
  expectation-maximization algorithm for fitting latent-variable graphical
  models to multivariate time series,'' \emph{Entropy}, vol.~20, p.~76, 01
  2018.

\bibitem{zorzi_reweight}
M.~Zorzi, ``Empirical {B}ayesian learning in {AR} graphical models,''
  \emph{Automatica}, vol. 109, 2019.

\bibitem{ZORZI2020109053}
M.~{Zorzi}, ``Autoregressive identification of {K}ronecker graphical models,''
  \emph{Automatica}, vol. 119, p. 109053, 2020.

\bibitem{v2020topology}
M.~S. Veedu, H.~Doddi, and M.~V. Salapaka, ``Topology learning of linear
  dynamical systems with latent nodes using matrix decomposition,'' 2020.

\bibitem{veedu2020topology}
M.~S. Veedu and M.~V. Salapaka, ``Topology identification under spatially
  correlated noise,'' 2020.

\bibitem{Tesi_Falconi}
\BIBentryALTinterwordspacing
L.~Falconi, ``Robust factor analysis of moving average processes,'' 2021,
  master thesis. [Online]. Available: \url{http://tesi.cab.unipd.it/65215/}
\BIBentrySTDinterwordspacing

\bibitem{songsiri2010graphical}
J.~Songsiri, J.~Dahl, and L.~Vandenberghe, ``Graphical models of autoregressive
  processes,'' \emph{Convex optimization in signal processing and
  communications}, pp. 89--116, 2010.

\bibitem{BoydADMM}
\BIBentryALTinterwordspacing
S.~Boyd, N.~Parikh, E.~Chu, B.~Peleato, and J.~Eckstein, ``Distributed
  optimization and statistical learning via the alternating direction method of
  multipliers,'' \emph{Found. Trends Mach. Learn.}, vol.~3, no.~1, pp. 1--122,
  Jan. 2011. [Online]. Available: \url{http://dx.doi.org/10.1561/2200000016}
\BIBentrySTDinterwordspacing

\bibitem{boyd:vandenberghe:2004}
S.~Boyd and L.~Vandenberghe, \emph{Convex optimization}.\hskip 1em plus 0.5em
  minus 0.4em\relax Cambridge, United Kingdom: Cambridge University Press,
  2004.

\bibitem{CDC2020_Crescente}
F.~{Crescente}, L.~{Falconi}, F.~{Rozzi}, A.~{Ferrante}, and M.~{Zorzi},
  ``Learning ar factor models,'' in \emph{2020 59th IEEE Conference on Decision
  and Control (CDC)}, 2020, pp. 274--279.

\bibitem{ZAMORAMARTINEZ2014}
\BIBentryALTinterwordspacing
F.~Zamora-Martínez, P.~Romeu, P.~Botella-Rocamora, and J.~Pardo, ``On-line
  learning of indoor temperature forecasting models towards energy
  efficiency,'' \emph{Energy and Buildings}, vol.~83, pp. 162 -- 172, 2014.
  [Online]. Available:
  \url{http://www.sciencedirect.com/science/article/pii/S0378778814003569}
\BIBentrySTDinterwordspacing

\end{thebibliography}

\end{document}